\documentclass[a4paper,11pt]{amsart}
\usepackage{amsmath,amssymb,stmaryrd}
\usepackage{amsfonts}
\usepackage{eucal}
\usepackage{amsthm}
\numberwithin{equation}{section}

\newcommand{\bigpare}[1]{\bigl(#1\bigr)}
\newcommand{\biggpare}[1]{\biggl(#1\biggr)}
\newcommand{\Bigpare}[1]{\Bigl(#1\Bigr)}

\newcommand{\bigbra}[1]{\bigl\{#1\bigr\}}

\newcommand{\bigbrac}[1]{\bigl[#1\bigr]}

\newcommand{\biggbrac}[1]{\biggl[#1\biggr]}

\newcommand{\normalset}[2]{\{#1\mid#2\}}
\newcommand{\bigset}[2]{\bigl\{#1\bigm|#2\bigr\}}

\newcommand{\bignorm}[1]{\bigl\| #1 \bigr\|}

\newcommand{\biggnorm}[1]{\biggl\| #1 \biggr\|}

\newcommand{\abs}[1]{| #1 |}
\newcommand{\bigabs}[1]{\bigl| #1 \bigr|}
\newcommand{\Bigabs}[1]{\Bigl| #1 \Bigr|}
\newcommand{\biggabs}[1]{\biggl| #1 \biggr|}

\newcommand{\jap}[1]{\langle #1 \rangle}
\newcommand{\bigjap}[1]{\bigl\langle #1 \bigr\rangle}

\def\a{\alpha}
\def\b{\beta}
\def\c{\gamma}
\def\d{\delta}
\def\e{\varepsilon}
\def\f{\varphi}
\def\g{\psi}

\def\i{\mbox{\raisebox{.5ex}{$\chi$}}}

\def\l{\lambda}
\def\m{\mu}
\def\n{\nu}

\def\s{\sigma}

\def\x{\xi}
\def\y{\eta}
\def\z{\zeta}

\newcommand{\F}{\Phi}

\renewcommand{\L}{\Lambda}
\renewcommand{\O}{\Omega}
\newcommand{\Th}{\Theta}

\def\re{\mathbb{R}}
\def\co{\mathbb{C}}
\def\ze{\mathbb{Z}}

\def\pa{\partial}

\newcommand{\trace}{\mathrm{Tr} }

\DeclareMathOperator*{\slim}{s-lim}

\newcommand{\Ran}{\mathrm{Ran\;}}

\newtheorem{thm}{Theorem}[section]
\newtheorem{lem}[thm]{Lemma}

\newtheorem{cor}[thm]{Corollary}

\theoremstyle{definition}

\newtheorem{ass}{Assumption}

\theoremstyle{remark}
\newtheorem{rem}{Remark}[section]

\numberwithin{equation}{section}

\title{Long-range Scattering Matrix for \\
Schr\"odinger-type Operators}
\author{Shu Nakamura}
\address{
Department of Mathematics, Gakushuin University, 
1-5-1, Mejiro, Toshima, Tokyo, Japan 171-8588}
\email{shu.nakamura@gakushuin.ac.jp}
\thanks{The work is partially supported by JSPS Grant Kiban-B 15H03622. 
It is partially inspired by discussions with Dimitri Yafaev during the author's staying at Isaac 
Newton Institute for Mathematical Sciences for the program: Periodic and Ergodic Spectral Problems, supported by 
EPSRC Grant Number EP/K032208/1. He thanks Professor Yafaev for the valuable discussion, and the institute 
and the Simons Foundation for the financial support and its hospitality.} 
\date{\today}

\begin{document}

\maketitle

\begin{abstract}
We show that the scattering matrix for a class of Schr\"odinger-type operators 
with long-range perturbations is a Fourier integral operator 
with the phase function which is the generating function of the modified classical 
scattering map. 
\end{abstract}

\section{Introduction}\label{section-introduction}

In this paper, we consider Schr\"odinger type-operator:
\[
H=H_0+V
\]
on $L^2(\re^d)$, $d\geq 1$.
The unperturbed operator $H_0$ has the form: 
$H_0= p_0(D_x)$ on $\re^d$, where $D_x=-i\pa_x$  
and $p_0$ is a real-valued smooth function. We suppose: 

\begin{ass}\label{ass-free-symbol}
Let $m>0$. 
$p_0\in S^{m}$, i.e., for any multi-index $\a\in\ze_+^d$, there is $C_\a>0$ such that 
\[
\bigabs{\pa_\x^\a p_0(\x)}\leq C_\a\jap{\x}^{m-|\a|}, \quad \x\in\re^d.
\]
Moreover, we suppose $p_0$ is elliptic, i.e., there is $c_0, c_1>0$ such that 
\[
\bigabs{p_0(\x)}\geq c_0\jap{\x}^m-c_1, \quad \x\in\re^d.
\]
\end{ass}

The perturbation term $V$ is a symmetric pseudodifferential operator, and satisfies the 
following assumption. We let  
\[
g=\frac{dx^2}{\jap{x}^2}+d\x^2
\]
be our standard metric on $T^*\re^d$. Then we use the H\"ormander $S(m,g)$-class 
notation with respect this metric, i.e., for a weight function $m(x,\x)$, $a\in S(m,g)$ if 
for any $\a,\b\in\ze_+^d$, there is $C_\a\b>0$ such that 
\[
\bigabs{\pa_x^\a\pa_\x^\b a(x,\x)}\leq C_{\a\b}m(x,\x)\jap{x}^{-|\a|}, \quad x,\x\in\re^d.
\]

\begin{ass}\label{ass-potential}
Let $m$ be as in Assumption~\ref{ass-free-symbol}. 
There is $\m\in (0,1)$ such that $V\in S(\jap{x}^{-\m}\jap{\x}^m, g)$, 
and $V$ is real-valued. 
\end{ass}

We denote the Weyl quantization of $V$ by the same symbol: 
$V=V^W(x,D_x)$, and $V$ is a symmetric operator on $L^2(\re^d)$. 

\begin{rem}
For $\m>1$, the perturbation is short-range, and the scattering theory is much simpler
(see \cite{N2016}).
Here we also exclude the case $\m=1$. This case is especially important, because 
the (mollified) Coulomb potential satisfies this assumption. Though such potential satisfies 
the condition with any $\m<1$, more precise results should hold, and we address this in a 
separate paper \cite{N2018-2}. 
\end{rem}

We denote the symbol of $H$ by $p_1$:
\[
p_1(x,\x)= p_0(\x)+V(x,\x), \quad x,\x\in\re^d.
\]
Since $V$ decays as $|x|\to\infty$, we have ellipticity of $p_1$ if $|x|$ is sufficiently large. 
However, since our perturbations include the metric perturbation, we need to assume global 
ellipticity: 

\begin{ass}\label{ass-ellipticity}
There is $c_2, c_3>0$ such that 
\[
|p_1(x,\x)|\geq c_2\jap{\x}^m, \quad \text{if }|\x|\geq c_3, x,\x\in\re^d.
\]
\end{ass}

Under these assumptions, it is well-known that $H$ is self-adjoint on $H^m(\re^d)$. 
We write the unique self-adjoint extension by $H$ as well as the pseudodifferential operator. 

We now fix an energy interval $I=[E_0,E_1]\Subset \re$, and we consider the scattering 
on $I$. We note, by Assumption~\ref{ass-potential}, 
\[
\O_I^0=\bigset{\x\in\re^d}{p_0(\x)\in I}\subset\re^d
\]
is bounded. We assume the following non-degenerate condition on the interval $I$: 

\begin{ass} \label{ass-energy-interval}
For $x\in \O_I^0$, $\pa_\x p_0(\x)\neq 0$.
\end{ass}

Under these assumptions, we can apply the Mourre theory with the conjugate operator 
$A=\frac12(x\cdot \pa_\x p_0(D_x)+\pa_\x p_0(D_x)\cdot x)$, and we learn the spectrum 
of $H$ on $I$ is absolutely continuous possibly except for finite number of eigenvalues
(see, e.g., \cite{ABM}, \cite{N2016}). 

Following Isozaki-Kitada \cite{Isozaki-Kitada1,Isozaki-Kitada2}, Derezi\'nski-G\'erard \cite{Derezinski-Gerard}
and Robert \cite{Robert}, 
we construct time-independent modifiers $J_\pm$
in our setting in Section~\ref{section-isozaki-kitada} (which depends on $I$). 
Using these we can define modified wave operators: 
\[
W_\pm^I = \slim_{t\to\pm\infty} e^{itH} J_\pm e^{-itH_0} E_I(H_0)
\]
Then the existence of these limits are proved by the same method as in the papers 
by Isozaki-Kitada \cite{Isozaki-Kitada1},  and $W_\pm$ are partial isometries 
on $\Ran[E_{I}(H_0)]$. Moreover, the asymptotic completeness is also proved 
by the standard method:
\[
\Ran [W_\pm^I E_{I}(H_0)] = E_{I}(H)\mathcal{H}_{c}(H),
\]
where $\mathcal{H}_{c}(H)$ is the continuous spectral subspace with respect to $H$. 
The scattering operator $S^I$ is defined by 
\[
S^I=(W_+^I)^* W_-^I,
\]
and it is an isometry on $\Ran[E_{I}(H_0)]$. It is well-known that $S^I$ commutes with 
the free Hamiltonian: $S^IH_0=H_0 S^I$. 

We then introduce the scattering matrix. We employ the formulation in 
Nakamura \cite{N2016}. For $\l\in I$, we set the energy surface $\Sigma_\l$ by  
\[
\Sigma_\l = p_0^{-1}(\{\l\}) = \bigset{\x\in\re^d}{p_0(\x)=\l}.
\]
We note $\Sigma_\l$ is a smooth submanifold for $\l\in I$ since $\pa_\x p_0(\x)\neq 0$ on $p_0^{-1}(I)$. 
We let a measure $m_\l$ on $\Sigma_\l$ defined by $m_\l(\x)= |\pa_x p_0(\x)|^{-1}dS_\x$ so that 
$m_\l\wedge dp_0(\x)=d\x$, where $dS_\x$ is the surface density (measure) on $\Sigma_\l$. 
Let $T(\l)$ be the trace operator from $H_{\mathrm{loc}}^s(\re^d_\x)$ ($s>1/2$) to $L^2(\Sigma_\l)$ defined by
\[
T(\l) \ :\ f\mapsto f\big|_{\Sigma_\l}\in L^2(\Sigma_\l), \quad f\in H_{\mathrm{loc}}^s(\re^d). 
\]
Then 
\[
T(\cdot)\ :\ f \mapsto (T(\l)f)\in \int^\oplus_I L^2(\Sigma_\l,m_\l)d\l
\]
is extended to a surjective partial isometry with the initial space $L^2(p_0^{-1}(I))$. 
In particular, $T(\cdot)\mathcal{F}$ is a spectral representation of $H_0$ on $\Ran [E_I(H_0)]$. 
Then $\mathcal{F} S^I\mathcal{F}^*$ is decomposed on this spectral representation space, 
and we have 
\[
T(\cdot)\mathcal{F} S^I \mathcal{F}^* T(\cdot)^* = \int_I^\oplus S(\l)d\l, 
\]
with $S(\l)\in \mathcal{B}(L^2(\Sigma_\l,m_\l))$. $S(\l)$ is the scattering matrix, and it is 
easy to show $S(\l)$ is unitary for (at least) almost all $\l\in I$. 

\begin{thm}\label{thm-main}
Let $\l\in I\setminus\s_{\mathrm{p}}(H)$. 
Then there are $\g(y,\y)\in S^1_{1,0}$ on $T^*\Sigma_\l$ and $a(y,\y)\in S^0_{1,0}$ 
such that 
\[
S(\l)\f(\y) =(2\pi)^{-(d-1)}\iint e^{-i\g(y,\y)+iy\cdot\z}\tilde\Theta(y,\y) a(y,\y) \f(\z)d\z dy
\]
for $\f\in C_0^\infty(\Sigma_\l)$ in a local coordinate of $\Sigma_\l$, where 
\[
\tilde\Theta(y,\y) = \bigabs{\det (\pa_y\pa_\y\g(y,\y))}^{1/2}. 
\]
Moreover, $\g(y,\y)-y\cdot\y\in S_{1,0}^{1-\m}$ and 
the principal symbol of $a(y,\y)$ is $1$, i.e., $a(y,\y)-1\in S^{-1}_{1,0}$. 
Here we have used the standard Kohn-Nirenberg symbol notation $S^m_{\rho,\delta}$, 
$m\in\re$, $0\leq\d\leq\rho\leq 1$. 
\end{thm}

\begin{rem}
Even though $S(\l)$ is not a pseudodifferential operator in general, it has pseudo-local 
property since $\pa_y\g(y,\y)-\y =O(\jap{y}^{-\m})$ as $|y|\to\infty$. 
Hence it is sufficient to consider such operators in a local coordinate, 
as well as pseudodifferential operators. This class of Fourier integral operators is 
somewhat different form the standard H\"ormander-type Fourier integral operators, 
where the phase functions are supposed to be homogeneous of order one with respect to 
the conjugate variables ($y$ in our setting). We note the calculus of Asada-Fujiwara
\cite{Asada-Fujiwara} still applies to our class of operators. 
\end{rem}

\begin{rem}
$\g(y,\y)$ is the generating function of classical scattering 
map, and we discuss the construction in detail later in this paper. 
The factor $\tilde\Theta(y,\y)$ corresponds to the modification of the volume form, and 
it makes the operator approximately unitary. 
\end{rem}

\begin{rem}
In principle, we can compute $\g(y,\y)$ explicitly in terms of classical mechanics. 
For many examples, at least if $\m>1/2$, we can compute the asymptotic behavior as $|y|\to\infty$
(see \cite{Yafaev-1998}). 
We can consider $\exp(-i\g(-D_\y,\y))$ as a good approximation of the scattering matrix, 
and hence we expect the spectral properties of $S(\l)$ is decided by the behavior. 
See Nakamura \cite{N2018-2} for the case $\m=1$. 
\end{rem}

The long-range scattering theory for Schr\"odinger operators has a long history, and 
there is substantial literature on this subject, especially for two-body case. 
We refer textbooks Reed-Simon \cite{Reed-Simon} \S 11.9, 
Derezi\'nski-G\'erard \cite{Derezinski-Gerard}, Yafaev \cite{Yafaev-SLNM1735} , \cite{Yafaev-AMS2009} 
and references therein. Long-range scattering for discrete Schr\"odinger operators has also been 
studied by several authors recently (\cite{N2014}, \cite{Tadano}), and in this paper we consider 
relatively large class of operators, of which the method is easily applied to discrete Schr\"odinger operators
as well. The literature on the scattering matrix for long-range scattering has been relatively few. 
The off-diagonal smoothness of the scattering matrix was proved by Isozaki-Kitada \cite{Isozaki-Kitada2}, 
and also studied by Yafaev (see \cite{Yafaev-SLNM1735} and references therein). 
The Fourier integral operator representation of the scattering matrix was studied by Yafaev 
in the case of $\m>1/2$ using the Dollard-type approximate solutions to the eikonal equation 
(\cite{Yafaev-1998}). 
In this paper we employ explicit construction of the solutions to eikonal equation with precise control 
of the classical trajectories and ideas from interaction pictures. 

Our argument relies heavily on the formulation of long-range scattering in terms of time-independent 
modifiers by Isozaki and Kitada (\cite{IK0,Isozaki-Kitada1,Isozaki-Kitada2,IK3}. See also alternative 
constructions by Robert \cite{Robert}, Derezi\'nski-G\'erard \S 4.15, Yafaev \cite{Yafaev-SLNM1735}). 
In this paper, we give relatively detailed analysis of the classical mechanics, partly because the 
system we consider is more general than the Schr\"odinger (or Newton) equation, but also 
because the settings and constructions are somewhat different, and the construction itself is 
important to understand the meaning of the representation. The other source of the proof 
of Theorem~\ref{thm-main} is a recent result by the author on the short-range scattering matrix 
\cite{N2016}, and we modify its argument to apply to the long-range case. 
Another recent result on microlocal resolvent estimates \cite{N2017} is also crucial in 
the proof (for our generalized system). 

The paper is constructed as folows: In Section~\ref{section-cl}, we prepare global-in-time 
estimates for the solutions to Hamilton equations with nontrapping dynamical system, 
and construct solutions to Hamilton-Jacobi equations and eikonal equations, using the 
idea of interaction picture. 
In Section~\ref{section-isozaki-kitada}, we construct the time-independent modifiers following 
the idea of Isozaki and Kitada. In Section~\ref{section-qm}, we give the proof of Theorem~\ref{thm-main}, 
using the idea of \cite{N2016}. We use microlocal analysis extensively in Section~\ref{section-isozaki-kitada} 
and Section~\ref{section-qm}, and we refer H\"ormander \cite{Hormander}, Sogge \cite{Sogge} 
and Asada-Fujiwara \cite{Asada-Fujiwara}.

\section{Preparation on classical mechanics} \label{section-cl}

\subsection{Classical mechanics with space cut-off}\label{subsection-cl-cutoff}

We introduce a constant $R>0$, and we set
\[
V_R(x,\x)= \i_1(|x|/R)V(x,\x), 
\]
where $\i_1\in C^\infty(\re)$ is a smooth cut off function such that $\i_1(s)=0$ if 
$s\leq 1$; $\i_1(s)=1$ if $s\geq 2$. We then set
\[
p(x,\x)= p_0(x,\x)+V_R(x,\x). 
\]
We fix $R$ later in this subsection. 
We now consider the classical mechanics generated by $p(x,\x)$. Namely, we consider 
solutions to the Hamilton equation:
\[
\frac{d}{dt} x(t) =\frac{\pa p}{\pa\x}(x(t),\x(t)), \quad 
\frac{d}{dt}\x(t) =-\frac{\pa p}{\pa x}(x(t),\x(t))
\]
with the initial condition: $x(0)=x_0$, $\x(0)=\x_0$. We denote the solution to the equation by 
\[
\exp t H_p(x_0,\x_0)= (x(x_0,\x_0;t), \x(x_0,\x_0;t))\in\re^{2d}, \quad x_0,\x_0\in\re^d, t\in\re. 
\]

We now recall Assumption~\ref{ass-energy-interval}, and we consider a trajectories with the energy 
$\l$ in a neighborhood of  $I=[E_0,E_1]$. We choose $\e_0>0$ so that there is $c_4>0$ such that 
\[
|\pa_\x p_0(\x)|\geq c_4, \quad \text{for }\x\in \O_{I_6}^0,
\]
where 
\[
I_k=[E_0-k\e_0,E_1+k\e_0], \quad k=0,1,2,\dots 6. 
\]
$\pa_\x p_0(\x)$ is the free velocity, and we denote it as 
\[
v(\x) =\pa_\x p_0(\x). 
\]
We denote the Poisson bracket of $a, b\in C^\infty(\re^{2d})$ by
\[
\{a,b\} =\sum_{j=1}^d \frac{\pa a}{\pa x_j}\frac{\pa b}{\pa\x_j} 
- \frac{\pa a}{\pa\x_j}\frac{\pa b}{\pa x_j}, 
\]
and we write the unit matrix on $\co^d$ by $\mathrm{E}$. We also denote
\[
\O_J =\bigset{(x,\x)}{p(x,\x)\in J}, \quad J\subset\re.
\]

\begin{lem}\label{lem-cl-convexity}
There is $R_0>0$ such that if $R\geq R_0$, then there is $c_5>0$ such that 
\[
\frac{d^2}{dt^2}|x(t)|^2 =\{\{|x|^2,p\},p\}(x,\x)\geq c_5, \quad x,\x\in\re^d,
\]
if $(x_0,\x_0)\in \O_{I_5}$. Moreover, for each $x_0$, $t\in\re$, $\x_0\mapsto \x(x_0,\x_0;t)$ is a 
diffeomorphism, and 
\[
\det\biggbrac{\frac{\pa \x}{\pa \x_0}(x_0,\x_0;t)}\geq 1/2
\]
for any $x_0,\x_0\in\re^d$ such that $(x_0,\x_0)\in\O_{I_5}$, $t\in\re$. 
\end{lem}

\begin{proof}
This is a variation of the so-called {\em classical Mourre estimate}, and we only sketch the idea. 
By Assumption~D, we learn that if $R_0$ is sufficiently large, 
\[
|V(x,\x)|\leq \frac12|p_0(\x)|+c_1\quad \text{for }|x|\geq R_0, \x\in\re^d. 
\]
Then we have, provided $R\geq R_0$, 
\[
|p_0(\x)|\leq 2 |p(x,\x)|+2 c_1= 2|p(x_0,\x_0)|+2c_1, 
\quad \text{for }(x_0,\x_0)\in\O_{I_5},
\]
and hence $|\x|\leq M$ with some $M>0$, uniformly on $\O_{I_6}$. 
We choose $R$ so large that 
\[
|V_R(x,\x)|\leq \e_0 \quad\text{for }(x_0,\x_0)\in\O_{I_5}
\]
holds. Then, if $(x,\x)\in \O_{I_5}$ then $p_0(\x)\in I_6$ and hence $|v(\x)|\geq c_4$. 

Now, by the direct computations, we learn 
\begin{align*}
\{\{|x|^2,p\},p\} &= 2|v(\x)|^2 +O(\jap{\x}^{2m-2}\jap{x}^{-\m}\i_{\{|x|\geq R\}})\\
& =2|v(\x)|^2 +O(R^{-\m}), \quad\text{on }\O_{I_5}, 
\end{align*}
and hence 
\[
\{\{|x|^2,p\},p\}\geq c_4^2>0, \quad \text{for }(x_0,\x_0)\in \O_{I_5}, 
\]
provided $R$ is chosen sufficiently large. 
This also implies, in particular, for any solution to the Hamilton equation, there 
is $t_0\in\re$ such that 
\[
|x(t)| \geq (|x(t_0)|^2+c_4(t-t_0)^2/2)^{1/2}\geq (c_4/2)^{1/2} |t-t_0|, \quad t\in\re.
\]
Hence we have 
\[
\biggabs{\frac{d}{dt} \x(t)}\leq \biggabs{\frac{\pa V_R}{\pa x}(x(t),\x(t))} 
\leq C(\jap{t-t_0}+ R)^{-\m-1},
\]
and this implies 
\[
\abs{\x(t)-\x(s)}\leq C\int_s^t (\jap{u-t_0}+R)^{-\m-1}du \leq C'R^{-\m}, 
\quad -\infty<s<t<\infty. 
\]
Similarly we can show
\begin{equation}\label{eq-potential-small}
\biggnorm{\frac{\pa \x(t)}{\pa\x_0}-E}_{\co^d\to\co^d}\leq CR^{-\m},
\end{equation}
and we conclude the last assertion by choosing $R$ large enough. 
\end{proof}

In the following, we suppose $R$ is large enough that the argument of the above proof is valid.


\subsection{Solution to the Hamilton-Jacobi equation}
\label{subsection-HJ}

We consider the Hamilton-Jacobi equation in $\x$-space: 
\[
\frac{\pa}{\pa t} \phi(t,\x)=p\biggpare{\frac{\pa\phi}{\pa\x}(t,\x),\x}, 
\quad \x\in\O_{I_4}^0, t\in\re, 
\]
with the initial condition 
\[
\phi(0,\x)=0, \quad \x\in\re^d.
\]

We write 
\[
\Lambda_t : \y \mapsto \x(0,\y;t), 
\]
Then, by Lemma~\ref{lem-cl-convexity},  $\Lambda_t$ is locally diffeomorphic, and diffeomorphism 
from $\O_{I_5}^0$ into $\O_{I_6}^0$, and the range contains $\O_{I_4}^0$ 
(note \eqref{eq-potential-small}). 
By the standard theory of Hamilton-Jacobi equation (see, e.g., Evans \cite{Evans} Chapter~3, 
Arnold \cite{Arnold} \S47), the solution is constructed as follows: We set
\[
u(t,\y) =\int_0^t \bigbra{p(x(0,\y;s),\x(0,\y;s))
-x(0,\y;s)\cdot \pa_x V_R(x(0,\y;s),\x(0,\y;s))} ds.
\]
If we set 
\[
\phi(t,\x)= u(t,\Lambda_t^{-1}(\x)), \quad \x\in \O_{I_4}^0, 
\]
then $\phi$ is the solution to the Hamilton-Jacobi equation. 
We will show $\phi(t,\x)$ satisfies suitable symbol properties in the following.
For simplicity, in this subsection we write 
\[
x(t)=x(t,\y)=x(0,\y;t), \quad \x(t)=\x(t,\y)=\x(0,\y;t).
\]
At first we recall that there are $c>0$ such that 
\[
c^{-1} |t| \leq |x(t,\y)|\leq c |t|,  \quad t\in \re
\]
for any $\y\in \O_{I_5}^0$ by Lemma~\ref{lem-cl-convexity}. 
Here we may suppose $V_R=0$ in a neighborhood of 0.

\begin{lem}\label{lem-cl-tr1}
For any $\a\in\ze_+^d$, there is $C_\a>0$ such that 
\[
\bigabs{\pa_\y^\a x(t,\y)}\leq C_\a\jap{t}, \quad
\bigabs{\pa_\y^\a \x(t,\y)}\leq C_\a, 
\]
uniformly in $t\in\re$.
\end{lem}

\begin{proof}
For $\a=0$, the claim is obvious. By differentiating the Hamilton equation, we have 
\begin{align*}
& \pa_t (\pa_\y x)= (\pa_\x\pa_x p)\pa_\y x + (\pa_\x\pa_\x p)\pa_\y\x, \\
& \pa_t (\pa_\y \x)= -(\pa_x\pa_x p)\pa_\y x - (\pa_x\pa_\x p)\pa_\y\x.
\end{align*}
We note 
\begin{align*}
&\pa_\x\pa_\x p(x(t),\x(t)) = \pa_\x\pa_\x p_0(x(t),\x(t)) +\pa_\x\pa_\x V_R(x(t),\x(t))= O(1), \\
&\pa_\x\pa_x p(x(t),\x(t)) = \pa_\x\pa_x V(x(t),\x(t))=O(\jap{t}^{-1-\m}), \\
&\pa_x\pa_x p(x(t),\x(t))= \pa_x\pa_x V(x(t),\x(t))= O(\jap{t}^{-2-\m}). 
\end{align*}
Using these, we learn 
\begin{align*}
&|\pa_t (\pa_\y x)|\leq C\jap{t}^{-1-\m}|\pa_\y x|+ C |\pa_\y\x|, \\
&|\pa_t (\pa_\y \x)| \leq C \jap{t}^{-2-\m}|\pa_\y x| +C\jap{t}^{-1-\m}|\pa_\y\x|, 
\end{align*}
and these imply
\begin{align*}
&\pa_t (\jap{t}^{-1-\m/2}|\pa_\y x|)\leq C\jap{t}^{-1-\m}(\jap{t}^{-1-\m/2}|\pa_\y x|)
+ C \jap{t}^{-1-\m/2}|\pa_\y\x|, \\
&\pa_t |\pa_\y \x| \leq C \jap{t}^{-1-\m/2}(\jap{t}^{-1-\m/2}|\pa_\y x|) 
+C\jap{t}^{-1-\m}|\pa_\y\x|. 
\end{align*}
Here $\pa_t$ should be considered in distribution sense, and we have used the fact: 
$\pa_t \jap{t}^{-1-\m/2}\leq 0$. Combining them, we have 
\[
\pa_t(\jap{t}^{-1-\m/2}|\pa_\y x|+|\pa_\y\x|) 
\leq 2C\jap{t}^{-1-\m/2} (\jap{t}^{-1-\m/2}|\pa_\y x|+|\pa_\y\x|), 
\]
with $(\jap{t}^{-1-\m/2}|\pa_\y x|+|\pa_\y\x|)|_{t=0}= 1$. 
Then by the Gronwall's inequality, we learn 
\[
\jap{t}^{-1-\m/2}|\pa_\y x|+|\pa_\y\x| \leq C'<\infty 
\]
for all $t\in\re$, since $\jap{t}^{-1-\m/2}$ is integrable in $t$. This implies 
$\pa_\y x=O(\jap{t}^{1+\m/2})$ and $\pa_\y\x=O(1)$. Substituting these to 
the above equations again to learn
$\pa_t (\pa_\y x)=O(1)$ and hence $\pa_\y x=O(\jap{t})$. 

For higher derivatives, we use induction in $|\a|$. Suppose the claim holds for 
$|\a|\leq N-1$, and suppose $|\a|=N\geq 2$ We note 
\begin{align*}
\pa_t(\pa_\y^\a x) &= \pa_\y^\a (\pa_\x p(x(t,\y),\x(t,\y)) )\\
&= (\pa_\x\pa_x  p) \pa_\y^\a x + (\pa_\x\pa_\x p)\pa_\y^\a\x \\
&\quad +\sum_* c_*(\pa_\x \pa_x^\b\pa_\x^\c p) \prod_{i=1}^{d} 
\biggpare{\prod_{j=1}^{\b_i}(\pa_\y^{\tilde\b(i,j)} x_i)\prod_{k=1}^{\c_i} (\pa_\y^{\tilde\c(i,k)}\x_i)}, 
\end{align*}
where the last sum is taken over $\b,\c, \tilde\b(i,j), \tilde\c(i,k)\in\ze_+^d$ such that 
$|\b+\c|\geq 2$; $\tilde\b(i,j)\neq 0$ for $i=1,\dots,d$, $j=1,\dots,\b_i$;  $\tilde\c(i,k)\neq 0$ 
for $i=1,\dots,d$, $k=1,\dots,\c_i$; and  
\[
\sum_{i=1}^d\biggpare{\sum_{j=1}^{\b_i}\tilde \b(i,j) +\sum_{k=1}^{\c_i}\tilde \c(i,k)} =\a.
\]
Here $c_*$ denote some universal constants depending only on the indices. 
By the fact $\pa_\x \pa_x^\b\pa_\x^\c p=O(\jap{t}^{-|\b|})$ and the induction hypothesis, 
we learn that the last term is $O(1)$. Then by the same argument as above, we have
\[
|\pa_t (\pa_\y^\a x)|\leq C\jap{t}^{-1-\m}|\pa_\y^\a x|+ C |\pa_\y^\a\x| +C, 
\]
and similarly 
\[
|\pa_t (\pa_\y^\a \x)| \leq C \jap{t}^{-2-\m}|\pa_\y^\a x| +C\jap{t}^{-1-\m}|\pa_\y^\a\x|
+C\jap{t}^{-1-\m}. 
\]
Combining them, we learn 
\[
\pa_t(\jap{t}^{-1-\m/2}|\pa_\y^\a x|+|\pa_\y^\a\x|+1) 
\leq 2C\jap{t}^{-1-\m/2} (\jap{t}^{-1-\m/2}|\pa_\y^\a x|+|\pa_\y^\a\x|+1).
\]
Again by Gronwall's inequality, we have $\jap{t}^{-1-\m/2}|\pa_\y^\a x|+|\pa_\y^\a\x|+1\leq C$, 
and hence $\pa_\y^\a x=O(\jap{t}^{1+\m/2})$, $\pa_\y^\a\x=O(1)$. We substitute 
these to the above inequality again to conclude $\pa_\y^\a x=O(\jap{t})$. 
\end{proof}

\begin{lem}\label{lem-cl-HJ1}
For any $\a\in \ze_+^d$, there is $C_\a>0$ such that 
\[
\bigabs{\pa_\x^\a \phi(t,\x)}\leq C_\a \jap{t}, \quad t\in\re, 
\]
uniformly in $\x\in \O_{I_4}^0$. 
\end{lem}

\begin{proof}
By direct computation using the definition of $u(t,\y)$ and the previous lemma, we learn 
\[
\pa_\y^\a u(t,\y)=O(\jap{t}), \quad t\in\re,
\]
for any $\a\in\ze_+^d$, uniformly in $\y\in \O_{I_5}^0$. 
We note, by Lemma~\ref{lem-cl-convexity} and the previous lemma, we learn 
\[
\pa_\x^\a \Lambda_t^{-1}(\x)=O(1), \quad t\in\re,
\]
also uniformly on the range of $\Lambda_t(\O_{I_5}^0)$. Combining these, we learn 
\[
\pa_\x^\a \phi(t,\x) =\pa_\x^\a (u\circ \Lambda_t^{-1})(\x)= O(\jap{t}), 
\quad t\in\re, 
\]
uniformly in $\x\in\O_{I_4}^0\subset \Lambda_t(\O_{I_5}^0)$. 
\end{proof}

\begin{lem}\label{lem-cl-HJ2}
For any $\a\in\ze_+^d$, there is $C_\a>0$ such that 
\[
\bigabs{\pa_\x^\a (\phi(t,\x)-tp_0(\x))}\leq C_\a \jap{t}^{1-\m}, 
\quad t\in\re, \x\in \O_{I_4}^0. 
\]
\end{lem}

\begin{proof}
We recall, by the construction of the solution to the Hamilton-Jacobi equation, 
$\pa_\x\phi(t,\x)=x(t,\Lambda_t^{-1}(\x))$, and hence 
\[
\bigabs{\pa_\x\phi(t,\x)} = \bigabs{x(t,\Lambda_t^{-1}(\x))}\geq c|t|, \quad t\in\re, 
\]
with some $c>0$, uniformly in $\x$. On the other hand, by the Hamilton-Jacobi equation, we have 
\begin{align*}
\phi(t,\x) &= \int_0^t p(\pa_\x\phi(s,\x),\x) ds 
= \int_0^t (p_0(\x)+V_R(\pa_\x\phi(s,\x),\x)ds \\
&=t p_0(\x) +\int_0^t V_R(\pa_\x\phi(s,\x),\x)ds, 
\end{align*}
and hence 
\begin{align*}
\bigabs{\phi(t,\x)-tp_0(\x)} &\leq \pm \int_0^t |V_R(\pa_\x\phi(s,\x),\x)|ds\\
&\leq \pm \int_0^t C \jap{s}^{-\m}ds 
\leq C\jap{t}^{1-\m}, \quad \text{for }\pm t\geq 0.
\end{align*}
For derivatives: $\pa_\x^\a (\phi(t,\x)-tp_0(\x))$, we differentiate the above 
equality, and we obtain the result using Lemma~\ref{lem-cl-tr1}.
\end{proof}

\subsection{Out-going/in-coming conditions}
\label{subsection-cl-ogic}

Throughout this section, we always suppose Assumptions~\ref{ass-free-symbol}--\ref{ass-energy-interval}, 
and consider classical trajectories with $p(x_0,\x_0)\in I_5\Subset \re$ as in Lemma~\ref{lem-cl-convexity}. 

Let $\b\in (-1,1]$, and consider the condition 
\begin{equation}\label{eq-ogic} 
\pm \cos(x_0,v(\x_0))= \pm\frac{x_0}{|x_0|}\cdot\frac{v(\x_0)}{|v(\x_0)|} \geq \b.
\end{equation}

\begin{lem}\label{lem-cl-ogic1}
Let $\b>-1$ be fixed. Then there are $c>0$ and $L\geq 0$ such that 
\[
|x(x_0,\x_0;t)|\geq c(\jap{x_0}+\jap{t}), \quad \pm t\geq 0, 
\]
respectively, provided $|x_0|\geq L$, and the condition \eqref{eq-ogic} is satisfied. 
\end{lem}

\begin{proof}
We may assume $\b<0$ without loss of generality, and we consider the case $t\geq 0$ only. 
The other case can be considered in the same way. 
We note, by the Hamilton equation, 
\[
\frac{d}{dt} |x(t)|^2 = 2x\cdot \pa_\x p(x,\x) = 2x\cdot v(\x) + 2x\cdot\pa_\x V_R(x,\x).
\]
By this equality, we easily observe that  if $\cos(x_0,v(\x_0))\geq |\b|$, then $|x(t)|^2\geq c\jap{t}^2$ 
for $t\geq 0$, provided $|x_0|$ is sufficiently large. Thus in the following we may assume 
\begin{equation}\label{eq-ogic-mod}
|\cos(x_0,v(\x_0))|\leq |\b|<1.
\end{equation}
Since $v(\x)$ and $\pa_\x V(x,\x)$ are bounded on $\O_{I_5}$, we have 
\[
\frac{d}{dt}|x(t)|^2\Big|_{t=0} \geq - C_1 |x_0|
\]
with a constant $C_1$. Then, by Lemma~\ref{lem-cl-convexity}, we learn 
\begin{align*}
|x(t)|^2 &\geq |x_0|^2  -C_1 |x_0|t + \frac{c_5}{2}t^2\\
&= |x_0|^2 + (c_5/4)t^2+ ((c_5/4)t -C_1|x_0|)t
\end{align*}
for $t\geq 0$. Hence, we have 
\[
|x(t)|^2\geq |x_0|^2 + (c_4/4) t^2, \quad \text{if } t\geq (4C_1/c_5)|x_0|. 
\]
Thus it suffices to consider the estimate for $t\in [0, C_2 |x_0|]$, where $C_2= 4C_1/c_5$. 

We now consider the {\em impact paramter}\/:
\[
y(t)= x(t) - (x(t)\cdot \hat v(\x(t)))\hat v(\x(t)), \quad t\in\re,
\]
where $\hat v(\x)=|v(\x)|^{-1}v(\x)$. 
We have 
\[
\frac{d}{dt}y(t) = \pa_\x V_R -(\pa_\x V_R\cdot \hat v)\hat v  -\biggpare{x\cdot \frac{d\hat v}{dt}}\hat v
-(x\cdot\hat v)\frac{d\hat v}{dt}. 
\]
We note 
\[
\frac{d\hat v}{dt}= \frac{d}{dt}\biggpare{\frac{v}{|v|}}
= \frac{1}{|v|}\frac{dv}{dt} -\biggpare{v\cdot\frac{dv}{dt}}\frac{v}{|v|^3}
\]
and 
\[
\frac{dv}{dt} = (\pa_\x\pa_\x p_0) \frac{d\x}{dt} 
=-(\pa_\x\pa_\x p_0)(\pa_x V_R) = O(\jap{x(t)}^{-1-\m}).
\]
We also note $|v|\geq c_4$ by Assumption~\ref{ass-energy-interval}. Thus we learn
\[
\biggabs{\frac{d}{dt} y(t)}\leq C_3 \jap{x(t)}^{-\m}\leq C_3 \jap{y(t)}^{-\m}
\]
with some constant $C_3>0$. We solve this differential inequality as follows: at first, we note
\[
\frac{d}{dt}\jap{y(t)}\geq -\biggabs{\frac{d}{dt}y(t)}\geq -C_3\jap{y(t)}^{-\m},
\]
and hence 
\[
\jap{y}^\m\frac{d}{dt}\jap{y} = \frac{1}{\m+1}\frac{d}{dt}\bigpare{\jap{y}^{\m+1}}\geq -C_3.
\]
Thus we have 
\[
\jap{y(t)}^{\m+1} \geq \jap{y(0)}^{\m+1}  -(\m+1)C_3 t, \quad t\geq 0.
\]
On the other hand, by the assumption \eqref{eq-ogic} and \eqref{eq-ogic-mod}, we have 
\[
|y(0)|^2 = |x_0|^2 - |x_0\cdot \hat v(\x_0)|^2 \geq (1-\b^2)|x_0|^2, 
\]
since $y(t)$ and $\hat v(\x)$ are perpendicular. Combining them, we have 
\begin{align*}
\jap{y(t)}^{\m+1} &\geq (1-\b^2)^{(\m+1)/2} |x_0|^{\m+1}-(\m+1)C_3t\\
&\geq (1-\b^2)^{(\m+1)/2}|x_0|^{\m+1}\bigpare{1 - C_4|x_0|^{-\m}}
\end{align*}
if $t\in [0,C_2|x_0|]$, where $C_4=(1-\b^2)^{-(\m+1)/2}(\m+1)C_3C_2$.
Now if we take $R>0$ so that $1-C_4R^{-\m}>1/2$, then 
\[
\jap{x(t)}\geq \jap{y(t)} \geq (1-\b^2)^{1/2}2^{-1/(\m+1)} |x_0|, \quad 
t\in[0, C_2|x_0|].
\]
This completes the proof. 
\end{proof}

We note that if $|x_0|\leq L$, then by Lemma~\ref{lem-cl-convexity} and its proof, we have 
\[
|x(x_0,\x_0;t)|\geq c\jap{t}-C, \quad t\in\re, 
\]
with some $c>0$ and $C>0$, uniformly. Thus the above estimates are always valid 
for sufficiently large $|t|$, and hence, under the above out-going/in-coming condition
\eqref{eq-ogic}, we have 
\[
|x(t)|\geq c\jap{x_0;t}, \quad \pm t\geq T,
\]
with some constant $c>0$ and sufficiently large $T>0$, uniformly in $({x_0},{\x_0})$, 
where we denote
\[
\jap{x;t}=(1+|x|^2+t^2)^{1/2}.
\]
We use the following notation: 
\begin{align*}
&\O_{J,\pm}(\b)=\bigset{(x,\x)}{p(x,\x)\in J, \pm\cos(x,v(\x))>\b}, \\
&\O_{J,\pm}^0(\b)=\bigset{(x,\x)}{x\in\re^d, p_0(\x)\in J, \pm\cos(x,v(\x))>\b}
\end{align*}
for $J\subset \re$ and $\b\in (-1,1]$. 

\begin{lem}\label{lem-cl-ogic2}
\begin{enumerate}
\renewcommand{\labelenumi}{{\rm (\roman{enumi})}}
\item There is a constant $C>0$ such that 
\begin{align*}
& |\pa_{x_0} x(t)|\leq C(1+\jap{{x_0}}^{-1-\m}|t|), \quad 
|\pa_{x_0} \x(t)|\leq C \jap{{x_0}}^{-1-\m}, \\
&|\pa_{\x_0} x(t)|\leq C|t|, \quad |\pa_{\x_0} \x(t)|\leq C,
\end{align*}
uniformly for $(x_0,\x_0)\in \O_{I_5,\pm}(\b)$. 
Moreover, 
\[
|\pa_{\x_0}\x(t)-\mathrm{E}|\leq C\jap{x_0}^{-\m}. 
\]
\item For $\a,\b\in\ze_+^d$, $|\a+\b|\geq 2$, there is $C_{\a\b}>0$ such that 
\[
|\pa_{x_0}^\a\pa_{\x_0}^\b x(t)|\leq C_{\a\b} \jap{{x_0}}^{-|\a|}|t|, 
\quad |\pa_{x_0}^\a \pa_{\x_0}^\b \x(t)|\leq C_{\a\b}\jap{{x_0}}^{-|\a|-\m},
\]
uniformly for $\pm t\geq 0$, $(x_0,\x_0)\in \O_{I_5,\pm}(\b)$. 
If, moreover, $\a\neq 0$, then 
\[
|\pa_{x_0}^\a\pa_{\x_0}^\b x(t)|\leq C_{\a\b} \jap{{x_0}}^{-|\a|-\m}|t|.
\]
\end{enumerate}
\end{lem}

\begin{proof}
We use an argument similar to the proof of Lemma~2.2, but with carefully controlling the 
dependence of the constants on $x_0$. 

(i) We first note, by the Hamilton equation, we have 
\begin{equation}\label{eq-cl-b1}
\left\{
\begin{aligned}
&\pa_t(\pa_{x_0} x)=(\pa_\x\pa_x p) \pa_{x_0}x +(\pa_\x\pa_\x p) \pa_{x_0}\x, \\
&\pa_t(\pa_{x_0} \x)=-(\pa_x\pa_x p) \pa_{x_0}x -(\pa_x\pa_\x p) \pa_{x_0}\x, \\
& (\pa_{x_0}x)(0)=\mathrm{E}, \quad (\pa_{x_0}\x)(0)=0.
\end{aligned}
\right. 
\end{equation}
We recall 
\[
|\pa_\x\pa_\x p|= O(1), \quad |\pa_x\pa_\x p|=O(\jap{x_0;t})^{-1-\m}), 
\quad |\pa_x\pa_x p|=O(\jap{x_0;t}^{-2-\m}),
\]
under our conditions. Thus we have 
\begin{align}
\pa_t|\pa_{x_0}x| &\leq C\jap{x_0;t}^{-1-\m}|\pa_{x_0} x| + C|\pa_{x_0}\x|,\label{eq-cl-b2}\\
\pa_t|\pa_{x_0}\x| &\leq C\jap{x_0;t}^{-2-\m}|\pa_{x_0} x| + C\jap{x_0;t}^{-1-\m}|\pa_{x_0}\x|. 
\label{eq-cl-b3}
\end{align}
Hence we have 
\[
\pa_t \bigpare{\jap{x_0;t}^{-1-\m/2}|\pa_{x_0}x|+|\pa_{x_0}\x|}
\leq C \jap{x_0;t}^{-1-\m/2} \bigpare{\jap{x_0;t}^{-1-\m/2}|\pa_{x_0}x|+|\pa_{x_0}\x|}. 
\]
Then by Gronwall's inequality, we learn 
\begin{align*}
\jap{x_0;t}^{-1-\m/2}|\pa_{x_0}x|+|\pa_{x_0}\x|
\leq C \bigpare{\jap{x_0;t}^{-1-\m/2}|\pa_{x_0}x|+|\pa_{x_0}\x|}_{t=0}
\leq C' \jap{x_0}^{-1-\m/2}
\end{align*}
uniformly, and hence 
\[
|\pa_{x_0}x(t)|\leq C\jap{x_0}^{-1-\m/2} \jap{x_0;t}^{1+\m/2}, \quad 
|\pa_{x_0}\x(t)| \leq C\jap{x_0}^{-1-\m/2}. 
\]
Substituting these to \eqref{eq-cl-b2} and \eqref{eq-cl-b3}, and we have 
\begin{align*}
\pa_t|\pa_{x_0}x| &\leq C\jap{x_0}^{-1-\m/2} \jap{x_0;t}^{-\m/2}+C\jap{x_0}^{-1-\m/2}
\leq C\jap{x_0}^{-1-\m/2}, \\
\pa_t|\pa_{x_0}\x| &\leq C \jap{x_0}^{-1-\m/2} \jap{x_0;t}^{-1-\m/2}, 
\end{align*}
and hence (using $\int_0^\infty \jap{x_0;t}^{-1-\m/2}dt \leq C\jap{x_0}^{-\m/2}$), 
\[
|\pa_{x_0}x|\leq C(1+\jap{x_0}^{-1-\m/2}|t| ), \quad 
|\pa_{x_0}\x|\leq  C\jap{x_0}^{-1-\m} .
\]
Iterating this procedure once more, we obtain 
\[
|\pa_{x_0}x|\leq C(1+\jap{x_0}^{-1-\m}|t| )
\]
Similarly, we have 
\begin{equation*}
\left\{
\begin{aligned}
&\pa_t(\pa_{\x_0} x)=(\pa_\x\pa_x p) \pa_{\x_0}x +(\pa_\x\pa_\x p) \pa_{\x_0}\x, \\
&\pa_t(\pa_{\x_0} \x)=-(\pa_x\pa_x p) \pa_{\x_0}x -(\pa_x\pa_\x p) \pa_{\x_0}\x, \\
& (\pa_{\x_0}x)(0)=0, \quad (\pa_{\x_0}\x)(0)=\mathrm{E},
\end{aligned}
\right. 
\end{equation*}
and hence  
\begin{align*}
\pa_t|\pa_{\x_0}x| &\leq C\jap{x_0;t}^{-1-\m}|\pa_{\x_0} x| + C|\pa_{\x_0}\x|,\\
\pa_t|\pa_{\x_0}\x| &\leq C\jap{x_0;t}^{-2-\m}|\pa_{\x_0} x| + C\jap{x_0;t}^{-1-\m}|\pa_{\x_0}\x|.
\end{align*}
We again obtain
\[
\pa_t \bigpare{\jap{x_0;t}^{-1-\m/2}|\pa_{\x_0}x|+|\pa_{\x_0}\x|}
\leq C \jap{x_0;t}^{-1-\m/2} \bigpare{\jap{x_0;t}^{-1-\m/2}|\pa_{\x_0}x|+|\pa_{\x_0}\x|},
\]
and by Gronwall's inequality, we obtain 
\[
\jap{x_0;t}^{-1-\m/2}|\pa_{x_0}x|+|\pa_{x_0}\x| \leq 
C \bigpare{\jap{x_0;t}^{-1-\m/2}|\pa_{\x_0}x|+|\pa_{\x_0}\x|}_{t=0}. 
\]
This implies 
\[
|\pa_{\x_0}x|\leq C \jap{x_0;t}^{1+\m/2}, \quad |\pa_{\x_0}\x| 
\leq C. 
\]
Substituting these to the above inequality for $\pa_{\x_0}x$ again to learn 
\[
\pa_t|\pa_{\x_0}x|\leq C \jap{x_0;t}^{-\m/2}+C\leq C',
\]
and hence $|\pa_{\x_0}x|=O(|t|)$. 
Also, by substituting these to the inequality for $\pa_{x_0}\x$, we learn 
\[
\pa_t|\pa_{\x_0}\x|\leq C \jap{x_0;t}^{-2-\m}\jap{t}+C\jap{x_0;t}^{-1-\m}
\leq C' \jap{x_0;t}^{-1-\m}, 
\]
and hence $|\pa_{\x_0}\x-\mathrm{E}| = O(\jap{x_0}^{-\m})$. 

(ii) We prove the claim by induction in $|\a+\b|$. We use a weaker induction hypothesis:
\begin{quote}
${(C_{k})}$: For $\a,\b\in \ze_+^d$, $|\a+\b|\leq k$, there is $C_{\a\b}>0$ such that 
\[
\bigabs{\pa_{x_0}^\a\pa_{\x_0}^\b x(t)}\leq C_{\a\b}\jap{x_0}^{-|\a|}\jap{x_0;t}, \quad 
\bigabs{\pa_{x_0}^\a\pa_{\x_0}^\b \x(t)}\leq C_{\a\b} \jap{x_0}^{-|\a|},
\]
uniformly for $t\geq 0$, $(x_0,\x_0)\in\O_{I_5}$ and $\cos(x_0,v(\x_0))\geq \b$. 
\end{quote}
We note $(C_0)$ is easy to show, and $(C_1)$ is already proved 
in (i) of the lemma. Suppose $(C_k)$ holds for $k< N$, and suppose $|\a+\b|=N\geq 2$. 
By the Hamilton equation and the Leibniz rule, we have 
\begin{align*}
\pa_t(\pa_{x_0}^\a\pa_{\x_0}^\b x) 
&= (\pa_\x\pa_x  p) \pa_{x_0}^\a\pa_{\x_0}^\b x + (\pa_\x\pa_\x p)\pa_{x_0}^\a\pa_{\x_0}^\b\x \\
&\quad  +\sum_* c_*(\pa_\x \pa_x^\c\pa_\x^\d p) \prod_{i=1}^{d} 
\biggpare{\prod_{j=1}^{\c_i}(\pa_{x_0}^{a(i,j)}\pa_{\x_0}^{b(i,j)} x_i)
\prod_{k=1}^{\d_i} (\pa_{x_0}^{\tilde a(i,k)}\pa_{\x_0}^{\tilde b(i,k)}\x_i)}, 
\end{align*}
where the last sum is taken over $\c,\d, a(i,j), \tilde a(i,j), b(i,k), \tilde b(i,k)\in\ze_+^d$ such that 
$|\c+\d|\geq 2$; $a(i,j)+b(i,j)\neq 0$ for $i=1,\dots,d$, $j=1,\dots,\c_i$;  
$\tilde a(i,k)+\tilde b(i,k)\neq 0$ for $i=1,\dots,d$, $k=1,\dots,\d_i$; and  
\[
\sum_{i=1}^d\biggpare{\sum_{j=1}^{\c_i}a(i,j) +\sum_{k=1}^{\d_i}\tilde a(i,k)} =\a, \quad 
\sum_{i=1}^d\biggpare{\sum_{j=1}^{\c_i}b(i,j) +\sum_{k=1}^{\d_i}\tilde b(i,k)} =\b.
\]
$c_*$ denotes suitable universal constant for each index. 
We denote the last term by $R_1$. 
Similarly, we have 
\begin{align*}
\pa_t(\pa_{x_0}^\a\pa_{\x_0}^\b \x) 
&= (\pa_x\pa_x  p) \pa_{x_0}^\a\pa_{\x_0}^\b x + (\pa_x\pa_\x p)\pa_{x_0}^\a\pa_{\x_0}^\b\x \\
& \quad +\sum_* c_*(\pa_x \pa_x^\c\pa_\x^\d p) \prod_{i=1}^{d} 
\biggpare{\prod_{j=1}^{\c_i}(\pa_{x_0}^{a(i,j)}\pa_{\x_0}^{b(i,j)} x_i)
\prod_{k=1}^{\d_i} (\pa_{x_0}^{\tilde a(i,k)}\pa_{\x_0}^{\tilde b(i,k)}\x_i)}. 
\end{align*}
We denote the last term by $R_2$. 

We now recall
\[
\bigabs{\pa_x^\c \pa_\x^\d p(x(t),\x(t))} \leq
C \jap{x_0;t}^{-|\c|}, 
\]
with some constant $C>0$. Combining this with the induction hypothesis, we learn 
\begin{align*}
&\biggabs{(\pa_\x \pa_x^\c\pa_\x^\d p) \prod_{i=1}^{d} 
\biggpare{\prod_{j=1}^{\c_i}(\pa_{x_0}^{a(i,j)}\pa_{\x_0}^{b(i,j)} x_i)
\prod_{k=1}^{\d_i} (\pa_{x_0}^{\tilde a(i,k)}\pa_{\x_0}^{\tilde b(i,k)}\x_i)}}\\
&\quad \leq 
C \jap{x_0;t}^{-|\c|} \jap{x_0;t}^{|\c|} \jap{x_0}^{-|\a|}=C\jap{x_0}^{-|\a|}, 
\end{align*}
and hence
\[
|R_1|\leq C\jap{x_0}^{-|\a|}.
\]
Similarly, we have 
\[
|R_2|\leq C\jap{x_0;t}^{-1-\m}\jap{x_0}^{-|\a|}.
\]
Thus we have 
\begin{align}
 \pa_t |\pa_{x_0}^\a\pa_{\x_0}^\b x| &\leq C\jap{x_0;t}^{-1-\m}|\pa_{x_0}^\a\pa_{\x_0}^\b x|
+C|\pa_{x_0}^\a\pa_{\x_0}^\b\x|+ C\jap{x_0}^{-|\a|},  \label{eq-cl-b4}\\
 \pa_t |\pa_{x_0}^\a\pa_{\x_0}^\b \x| &\leq C\jap{x_0;t}^{-2-\m}|\pa_{x_0}^\a\pa_{\x_0}^\b x|
+C\jap{x_0;t}^{-1-\m}|\pa_{x_0}^\a\pa_{\x_0}^\b\x| \nonumber \label{eq-cl-b5}\\
&\quad + C\jap{x_0;t}^{-1-\m}\jap{x_0}^{-|\a|}. 
\end{align}
Combining these, we obtain 
\begin{align*}
& \pa_t \Bigpare{\jap{x_0;t}^{-1-\m/2}|\pa_{x_0}^\a\pa_{\x_0}^\b x|
+|\pa_{x_0}^\a\pa_{\x_0}^\b\x|} \\
&\quad\leq C \jap{x_0;t}^{-1-\m/2} \Bigpare{\jap{x_0;t}^{-1-\m/2}|\pa_{x_0}^\a\pa_{\x_0}^\b x|
+|\pa_{x_0}^\a\pa_{\x_0}^\b\x|}
+C \jap{x_0;t}^{-1-\m/2}\jap{x_0}^{-|\a|}
\end{align*}
and 
\[
\Bigpare{\jap{x_0;t}^{-1-\m/2}|\pa_{x_0}^\a\pa_{\x_0}^\b x|
+|\pa_{x_0}^\a\pa_{\x_0}^\b\x| }\Big|_{t=0} = 0
\]
since $|\a+\b|\geq 2$. Then by Gronwall's inequality we have 
\[
\jap{x_0;t}^{-1-\m/2}|\pa_{x_0}^\a\pa_{\x_0}^\b x|
+|\pa_{x_0}^\a\pa_{\x_0}^\b\x| 
\leq C\jap{x_0}^{-|\a|-\m/2},
\]
and hence 
\[
|\pa_{x_0}^\a\pa_{\x_0}^\b x|\leq C\jap{x_0;t}^{1+\m/2}\jap{x_0}^{-|\a|}, \quad 
|\pa_{x_0}^\a\pa_{\x_0}^\b\x|\leq C\jap{x_0}^{-|\a|-\m/2}\leq C\jap{x_0}^{-|\a|}.
\]
Now we substitute these to \eqref{eq-cl-b4} to learn 
\[
\pa_t |\pa_{x_0}^\a\pa_{\x_0}^\b x| \leq C\jap{x_0;t}^{-\m/2}\jap{x_0}^{-|\a|}
+2C\jap{x_0}^{-|\a|} \leq 3C\jap{x_0}^{-|\a|}.
\]
Integrating this, we conclude 
\begin{equation}\label{eq-cl-b6}
|\pa_{x_0}^\a\pa_{\x_0}^\b x| \leq C|t|\jap{x_0}^{-|\a|}\leq C\jap{x_0;t}\jap{x_0}^{-|\a|}.
\end{equation}
In particular, we have proved the induction step $(C_k)$ with $k=N$, and thus $(C_k)$ holds 
for all $k\geq 0$. We substitute $(C_N)$ to \eqref{eq-cl-b5}, and we learn 
\[
\pa_t |\pa_{x_0}^\a\pa_{\x_0}^\b \x| \leq C\jap{x_0;t}^{-1-\m}\jap{x_0}^{-|\a|},
\]
and then by integrating in $t$, we have 
\begin{equation}\label{eq-cl-b7}
|\pa_{x_0}^\a\pa_{\x_0}^\b\x|\leq C\jap{x_0}^{-|\a|-\m}.
\end{equation}

Now we suppose $\a\neq 0$, and consider each term in $R_1$ more carefully: 
\[
r_*=(\pa_\x \pa_x^\c\pa_\x^\d p) \prod_{i=1}^{d} 
\biggpare{\prod_{j=1}^{\c_i}(\pa_{x_0}^{a(i,j)}\pa_{\x_0}^{b(i,j)} x_i)
\prod_{k=1}^{\d_i} (\pa_{x_0}^{\tilde a(i,k)}\pa_{\x_0}^{\tilde b(i,k)}\x_i)}.
\]
If $\c= 0$, then $r_*$ contains derivatives of $\x$ in $x_0$, and thus we learn 
$r_*= O(\jap{x_0}^{-|\a|-\m})$, by virtue of \eqref{eq-cl-b7} and (i) of the lemma. If $\c\neq 0$, then 
$\pa_\x \pa_x^\c\pa_\x^\d p=O(\jap{x_0;t}^{-|\c|-\m})$, and we also improve the 
estimate to obtain $r_*=O(\jap{x_0}^{-|\a|-\m})$. Thus, if $\a\neq 0$, we have 
$R_1=O(\jap{x_0}^{-|\a|-\m})$, and we obtain, instead of \eqref{eq-cl-b4}, 
\[
 \pa_t |\pa_{x_0}^\a\pa_{\x_0}^\b x| \leq C\jap{x_0;t}^{-1-\m}|\pa_{x_0}^\a\pa_{\x_0}^\b x|
+C|\pa_{x_0}^\a\pa_{\x_0}^\b\x|+ C\jap{x_0}^{-|\a|-\m}.
\]
Now we substitute \eqref{eq-cl-b6} and \eqref{eq-cl-b7} to this inequality to learn 
\[
 \pa_t |\pa_{x_0}^\a\pa_{\x_0}^\b x| \leq C\jap{x_0;t}^{-\m}\jap{x_0}^{-|\a|}
 +C\jap{x_0}^{-|\a|-\m}\leq 2C\jap{x_0}^{-|\a|-\m}.
 \]
 Integrating this in $t$, we conclude $|\pa_{x_0}^\a\pa_{\x_0}^\b x|=O(\jap{x_0}^{-|\a|-\m}|t|)$. 
\end{proof}


\subsection{Classical mechanics in the interaction picture}
\label{subsection-cl-ip}

We now consider the evolution of 
\[
y(t)=x(x_0,\x_0;t)-\pa_\x \phi(t,\x(x_0,\x_0;t)), \quad \x(t)= \x(x_0,\x_0;t),
\]
where $\phi(t,x)$ is the solution to the Hamilton-Jacobi equation constructed in 
Subsection~\ref{subsection-HJ}. The classical Hamiltonian for the evolution is given by 
\begin{align*}
q(t,y,\x)&= p(y+\pa_\x\phi(t,\x),\x) - p(\pa_\x\phi(t,\x),\x)\\
&=V_R(y+\pa_\x\phi(t,\x),\x) - V_R(\pa_\x\phi(t,\x),\x).
\end{align*}
For the completeness, we verify that $q(t,y,\x)$ generate the evolution:

\begin{lem}\label{lem-cl-ipevo}
Let $y(t)$, $\x(t)$, $q(t,y,\x)$ as above. Then 
\begin{align*}
&\frac{d}{dt} y(t) = \pa_\x q(t,y(t),\x(t)), \quad \frac{d}{dt}\x(t) = -\pa_y q(t,y(t),\x(t)), \\
&y(0)=x_0,\quad \x(0)=\x_0.
\end{align*}
\end{lem}

\begin{proof}
For $j=1,\dots,d$, we compute 
\begin{align*}
\frac{d}{dt}y_j(t)  &= \frac{d}{dt}x_j(t) -(\pa_{\x_j}\pa_t\phi)(t,\x(t)) 
- \sum_k (\pa_{\x_k} \pa_{\x_j} \phi)(t,\x(t))\frac{d}{dt}\x_k(t) \\
&= (\pa_{\x_j}p)(y+\pa_\x\phi,\x) + \sum_k (\pa_{\x_j} \pa_{\x_k} \phi)(t,\x(t))
(\pa_{x_k}p)(y+\pa_\x\phi,\x) -\pa_{\x_j} (p(\pa_\x\phi(t,\x),\x))\\
&=\pa_{\x_j}\bigpare{p(y+\pa_\x\phi(t,\x),\x)} - \pa_{\x_j} (p(\pa_\x\phi(t,\x),\x))\\
&=\pa_{\x_j} q(t,y,\x).
\end{align*}
Similarly, we have 
\begin{align*}
\frac{d}{dt}\x_j(t) &= -(\pa_{x_j} p)(y+\pa_\x\phi(t,\x),\x) \\
&= -\pa_y \bigpare{p(y+\pa_\x\phi(t,\x),\x)}= -\pa_y q(t,y,\x).
\end{align*}
The initial condition is obvious from the definition. 
\end{proof}

The existence of the classical long-range scattering is well-known, but we write 
it down for the completeness, and also for the later reference. 

\begin{lem}\label{lem-cl-lrscat}
For each $(x_0,\x_0)\in \O_{I_5}$, the limits
\[
x_\pm =\lim_{t\to\pm\infty} y(t), \quad \x_\pm=\lim_{t\to\pm\infty}\x(t)
\]
exist, and $\x_\pm\in \O_{I_4}^0$. 
\end{lem}

\begin{proof}
We recall, by Lemma~\ref{lem-cl-convexity}, we have 
\[
|x(t)|\geq c|t|-c', \quad t\in\re,
\]
with some constants $c,c'>0$, and $\x(t)$ is uniformly bounded. 
Now we observe 
\[
|\pa_t \x(t)|\leq |\pa_x V_R(x(t),\x(t))|\leq C\jap{t}^{-1-\m}, \quad t\in\re, 
\]
and hence the limit 
\[
\lim_{t\to\pm\infty} \x(t) = \x_0-\int_0^{\pm\infty} \pa_x V_R(x(t),\x(t))dt
\]
exist. This also implies 
\[
|\x(t)-\x_\pm|\leq C\jap{t}^{-\m}, \quad \pm t>0.
\]

We note 
\begin{align*}
\frac{d}{dt}y(t) &= \pa_\x q(t;y,\x) \\
&= (\pa_\x V_R)(y+\pa_\x\phi(t,\x),\x) - (\pa_\x V_R)(\pa_\x\phi(t,\x),\x)\\
&\qquad +(\pa_\x\pa_\x\phi(t,\x)) \bigpare{\pa_x V_R(y+\pa_\x\phi(t,\x))-\pa_x V_R(\pa_\x\phi(t,\x))}.
\end{align*}
We also note $y(t)+\pa_\x\phi(t,\x)= x(t)$, and hence each term can be bounded using 
Lemma~\ref{lem-cl-convexity},  and we learn $\frac{d}{dt}y(t)=O(\jap{t}^{-\m})$. 
Integrating this in $t$, we have 
$y(t)=O(\jap{t}^{1-\m})$. 
In particular, by Lemma~\ref{lem-cl-HJ2} we have 
\[
\bigabs{s y(t)+\pa_\x\phi(t,\x(t))}\geq c\jap{t}, \quad |t|\gg 0, \ c>0, 
\]
uniformly for $s\in [0,1]$. 

Now by using 
\begin{align}
\frac{d}{dt}y_j(t) &= \frac{\pa}{\pa \x_j}\Bigpare{V_R(y+\pa_\x\phi(t,\x),\x) - V_R(\pa_\x\phi(t,\x),\x)}
\nonumber\\
&= \sum_k\int_0^1 y_k(\pa_{x_k} \pa_{\x_j} V_R)(sy+\pa_\x\phi(t,\x),\x)ds\nonumber\\
&\qquad +\sum_{k,\ell} \int_0^1 y_k (\pa_{\x_j}\pa_{\x_\ell}\phi)(t,\x) 
(\pa_{x_k}\pa_{x_\ell}V_R)(sy+\pa_\x\phi(t,\x),\x)ds, \label{eq-cl-c1}
\end{align}
we obtain 
\[
\frac{d}{dt}|y(t)|\leq C|y(t)|\jap{t}^{-1-\m}, \quad t\in\re.
\]
By Gronwall's inequality, we learn $|y(t)|$ is uniformly bounded. 
By substituting this boundedness to the above equation \eqref{eq-cl-c1} again, we have 
$\frac{d}{dt}y(t)=O(\jap{t}^{-\m-1})$, and hence $\frac{d}{dt}y(t)$ is integrable. 
This implies the convergence of $y(t)$ as $t\to\pm\infty$. 
\end{proof}

Now we consider uniform estimates for out-going/in-coming initial conditions. 
We first prepare a preliminary lemma. 

\begin{lem}\label{lem-cl-iptr1}
Let $\b\in (-1,1]$, and suppose $(x_0,\x_0)\in \O_{I_5,\pm}(\b)$. Then 
\begin{enumerate}
\renewcommand{\labelenumi}{{\rm (\roman{enumi})}}
\item There is $C_0>0$ such that 
\[
|\x(x_0,\x_0;t)-\x_0|\leq C_0\jap{x_0}^{-\m}, \quad \pm t\geq 0,
\]
uniformly in $(x_0,\x_0)$. In particular, $|\x_\pm-\x_0|=O(\jap{x_0}^{-\m})$. 
\item There is $C_1>0$ such that 
\[
|y(t)|\leq C_1\jap{x_0}, \quad \pm t>0.
\]
\end{enumerate}
\end{lem}

\begin{proof}
We modify the proof of Lemma~\ref{lem-cl-lrscat}. By Lemma~\ref{lem-cl-ogic1}, we have 
\[
|\pa_t \x(t)| = |(\pa_x V_R)(x(t),\x(t))|\leq C\jap{x_0;t}^{-1-\m}, \quad t\in\re, 
(x_0,\x_0)\in \O_{I_5}. 
\]
By integrating this inequality in $t$, we conclude (i). 

In order to prove the claim (ii), we note that the statement is obvious for $x_0$ in a 
bounded set by virtue of Lemma~\ref{lem-cl-lrscat}. Hence, it suffices to consider it for $|x_0|\gg 0$. 
We also consider the case ``$+$'' only, since the other case is handled similarly. 
By the same argument as in the proof of Lemma~\ref{lem-cl-lrscat}, we have 
\begin{equation*}
|y(t)-x_0|=
|x(t)-x_0-\pa_\x\phi(t,\x(t))|\leq C_2\jap{t}^{1-\m}, \quad t\geq 0,
\end{equation*}
uniformly in $(x_0,\x_0)\in \O_{I_5,\pm}(\b)$. 
We recall that there is $c>0$ such that 
\[
|\pa_\x\phi(t,\x)|\geq c_0|t|, \quad t\in\re. 
\]
We choose $T_0>0$ so large that 
\[
C_2 \jap{T_0}^{1-\m}\leq (c_0/4)T_0.
\]
For the moment, we suppose $|x_0|\geq (c_0/4)T_0$ so that 
$t\geq (4/c_0)|x_0|$ implies $t\geq T_0$. 
Thus, if $t\geq (4/c_0)|x_0|$, then 
\begin{align*}
|sy(t)+\pa_\x\phi(t,\x)|&\geq |\pa_\x\phi(t,\x)|-|y(t)|\\
&\geq c_0|t|-(|x_0|+C_2\jap{t}^{1-\m}) \geq (c_0/2)|t|, 
\end{align*}
uniformly in $s\in[0,1]$. Hence, by \eqref{eq-cl-c1}, we learn
\begin{equation}\label{eq-cl-c2}
\frac{d}{dt}|y(t)|\leq C|y(t)|\jap{t}^{-1-\m}, \quad t\geq (4/c_0)|x_0|.
\end{equation}
On the other hand, since $\pa_\x q(t,y,\x)$ is bounded on 
$\tilde\O_{I_5}(t)=\normalset{(x+\pa_\x\phi(t,\x),\x)}{(x,\x)\in\O_{I_5}}$, we learn 
\[
|y(t)|\leq |x_0|+ C|t|, \quad t\in \re. 
\]
In particular, 
\[
|y(t)|\leq C'|x_0|, \quad t\in[0,(4/c_0)|x_0|].
\]
We now use Gronwall's inequality for  \eqref{eq-cl-c2} with the initial condition at $t=(4/c_0)|x_0|$ 
to conclude $|y(t)|\leq C|x_0|$ for all $t\geq 0$. 
\end{proof}

We now consider the derivatives of the evolution $(y(t),\x(t))$, i.e., 
for $\a,\b\in\ze_+^d$, we study the properties of $\pa_{x_0}^\a\pa_{\x_0}^\b y(t)$ 
and $\pa_{x_0}^\a\pa_{\x_0}^\b \x(t)$. 
We first prepare properties of $q(t,y,\x)$ along the classical trajectories. 

\begin{lem}\label{lem-cl-ipsym1}
Let $x_0$, $\x_0$, $y(t)$, $\x(t)$ as in Lemma~\ref{lem-cl-iptr1}. Then 
\begin{enumerate}
\renewcommand{\labelenumi}{{\rm (\roman{enumi})}}
\item For any $\b\in\ze_+^d$, there is $C_\b>0$ such that 
\[
\bigabs{(\pa_\x^\b q)(t;y(t),\x(t))}\leq C_\b \jap{x_0}\jap{x_0;t}^{-1}\jap{t}^{-\m},
\quad \pm t\geq 0.
\]
\item For any $\a,\b\in\ze_+^d$ with $\a\neq 0$, there is $C_{\a\b}>0$ such that 
\[
\bigabs{(\pa_y^\a \pa_\x^\b q)(t;y(t),\x(t))}\leq C_{\a\b} 
\jap{x_0;t}^{-\m-|\a|},
\quad \pm t\geq 0.
\]
\end{enumerate}
\end{lem}

\begin{proof}
By direct computations, we have 
\begin{align}
\pa_\x^\b q(t,y,\x) 
&= \sum_* c_* \Bigpare{(\pa_y^\c \pa_\x^\d V_R)(y+\pa_\x\phi(t,\x),\x) 
-(\pa_y^\c \pa_\x^\d V_R)(\pa_\x\phi(t,\x),\x)} \times \nonumber \\
&\qquad \times \prod_{i=1}^d \prod_{j=1}^{\c_i} (\pa_\x^{\tilde\b(i,j)}\phi)(t,\x), 
\label{eq-cl-c3}
\end{align}
where the indices runs over $\c,\d,\tilde\b(i,j)\in \ze_+^d$, $i=1,\dots,d$, 
$j=1,\dots,\c_i$, such that 
$\tilde\b(i,j)\neq 0$ for $i=1,\dots,d$, $j=1,\dots,\d_i$, and 
\[
\sum_{i=1}^d\sum_{j=1}^{\c_i} \tilde\b(i,j)+\d =\b, 
\]
and $c_*$ denote suitable universal constants. 

If $\a\neq 0$, then we also have 
\[
\pa_y^\a\pa_\x^\b q(t,y,\x) 
= \sum_* c_* (\pa_y^{\a+\c} \pa_\x^\d V_R)(y+\pa_\x\phi(t,\x),\x) 
\prod_{i=1}^d \prod_{j=1}^{\c_i} (\pa_\x^{\tilde\b(i,j)}\phi)(t,\x), 
\]
with the same set of indices. If $\a\neq 0$, then the claim (ii)  follows easily from the above 
expression combined with Lemmas~\ref{lem-cl-HJ1}, \ref{lem-cl-ogic1}. It remains to show (i). 
We first consider the case $\a=\b=0$, i.e., bounds on $q(t,y,\x)$ itself. 
We consider the ``$+$'' case only. 

By Lemma~\ref{lem-cl-HJ2}, there is $c_0>0$ such that 
$|\pa_\x\phi(t,\x)|\geq c_0|t|$, 
uniformly in $t$ and $\x_0\in \O_{I_3}^0$. Also, let $C_1$ as in the previous lemma, i.e., 
$|y(t)|\leq C_1 |x_0|$ for any $t\in\re$. If $t\geq M|x_0|$ with $M=2C_1/c_0$, 
then 
\[
\bigabs{sy(t)+\pa_\x\phi(t,\x)} 
\geq |\pa_\x\phi(t,x)| - |y(t)| 
\geq c_0|t| - C_1|x_0| \geq (c_0/2)|t|,
\]
uniformly for $s\in[0,1]$. Combining this with
\[
q(t,y,\x) =\int_0^1 y(t)\cdot (\pa_x V_R)(sy(t)+\pa_\x\phi(t,\x),\x) ds,
\]
we learn 
\[
|q(t,y(t),\x(t))|\leq C|x_0|\jap{t}^{-1-\m} \leq C'|x_0|\jap{x_0;t}^{-1-\m},
\]
if $t>M|x_0|$. 
In the last inequality, we have used $|x_0|^2+|t|^2\leq (1+M^2)|t|^2$ and hence 
$\jap{t}^{-1}\leq (1+M^2)^{1/2}\jap{x_0;t}^{-1}$. 
On the other hand, if $0\leq t \leq M|x_0|$, then we have 
\begin{align*}
|q(t,y(t),\x(t))|&\leq C\jap{x(t)}^{-\m} +C\jap{t}^{-\m} \\
&\leq C\jap{x_0;t}\jap{x_0;t}^{-1}\jap{t}^{-\m}
\leq C' \jap{x_0}\jap{x_0;t}^{-1}\jap{t}^{-\m}.
\end{align*}
Combining these, we conclude 
\[
|q(t,y(t),\x(t))|\leq C\jap{x_0}\jap{x_0;t}^{-1}\jap{t}^{-\m}, 
\quad t\geq 0.
\]
Estimates for $\pa_\x^\a q(t,y,\x)$ is similar, though more terms are involved. 
Actually, for $t\geq M|x_0|$, we have, using \eqref{eq-cl-c3}, 
\[
|\pa_\x^\b q| \leq C|x_0| \sum_{j=1}^{|\b|} \jap{t}^{-1-\m-j} \jap{t}^j
\leq C|x_0|\jap{t}^{-1-\m}
\]
as before, where $j$ corresponds to $|\c|$ in \eqref{eq-cl-c3}. For $t\in [0, M|x_0|]$, we similarly have 
\begin{align*}
|\pa_\x^\b q| 
&\leq C \sum_{j=1}^{|\b|} \bigbra{\jap{x(t)}^{-\m-j}+\jap{t}^{-\m-j}} \jap{t}^j\\
&\leq C(\jap{x_0;t}^{-\m}+\jap{t}^{-\m}) \leq C'\jap{x_0}\jap{x_0;t}^{-1}\jap{t}^{-\m}
\end{align*}
as above. 
\end{proof}

In the following, the next combined estimates are sometimes useful. 
\begin{cor}\label{cor-cl-ipsym2}
Let $x_0$, $\x_0$, $y(t)$, $\x(t)$ as in Lemma~\ref{lem-cl-iptr1}. 
Then for any $\a,\b\in\ze_+^d$, there is $C_{\a>\b}0$ such that 
\[
\bigabs{(\pa_y^\a \pa_\x^\b q)(t;y(t),\x(t))}\leq C_{\a\b} \jap{x_0}^{1-|\a|}\jap{x_0;t}^{-1}\jap{t}^{-\m},
\quad \pm t\geq 0.
\]
\end{cor}

\begin{proof}
If $\a=0$, the claim is the same as (i) in Lemma~\ref{lem-cl-ipsym1}. If $\a\neq 0$, then 
\[
\jap{x_0;t}^{-\m-|\a|} = \jap{x_0;t}^{1-|\a|} \jap{x_0;t}^{-1}\jap{x_0;t}^{-\m}
\leq \jap{x_0}^{1-|\a|}\jap{x_0;t}^{-1}\jap{t}^{-\m},
\]
since $1-|\a|\leq 0$, $\m>0$. Thus the claim follows from (ii) of Lemma~\ref{lem-cl-ipsym1}.
\end{proof}

We also use the following elementary estimate repeatedly. 

\begin{lem}\label{lem-cl-elementary}
Let $0<\m<1$. Then there is $C>0$ such that 
\[
\int_0^\infty \jap{a;t}^{-1}\jap{t}^{-\m}dt \leq C\jap{a}^{-\m}, \quad a\geq 0. 
\]
\end{lem}

\begin{proof}
Suppose $a\geq 1$. Then 
\begin{align*}
\int_0^\infty \jap{a;t}^{-1}\jap{t}^{-\m} dt
&\leq \sqrt{2} \int_0^\infty (a+t)^{-1}t^{-\m}dt \\
&= \sqrt{2} \int_0^\infty (1+s)^{-1}(as)^{-\m} ds\\
& = ca^{-\m}\leq 2^{\m/2} c\jap{a}^{-\m}
\end{align*}
with a constant $c>0$. If $0<a\leq 1$, then 
\[
\int_0^\infty \jap{a;t}^{-1}\jap{t}^{-\m}dt \leq \int_0^\infty\jap{t}^{-1-\m}dt \leq C
\leq \sqrt{2} C \jap{a}^{-\m}. 
\]
\end{proof}

\begin{lem} \label{lem-cl-iptr2}
\begin{enumerate}
\renewcommand{\labelenumi}{{\rm (\roman{enumi})}}
\item There is a constant $C>0$ such that 
\begin{align*}
& |\pa_{x_0} y(t)|\leq C, \quad 
|\pa_{x_0} \x(t)|\leq C \jap{{x_0}}^{-1-\m}, \\
&|\pa_{\x_0} y(t)|\leq C\jap{x_0}^{1-\m}, \quad |\pa_{\x_0} \x(t)|\leq C,
\end{align*}
uniformly for $\pm t\geq 0$, $({x_0},{\x_0})\in \O_{I_5,\pm}(\b)$. 
Moreover, 
\[
|\pa_{x_0}y(t)-\mathrm{E}|\leq C\jap{x_0}^{-\m}, 
\quad |\pa_{\x_0}\x(t)-\mathrm{E}|\leq C\jap{x_0}^{-\m}.
\]
\item For $\a,\b\in\ze_+^d$, $|\a+\b|\geq 2$, there is $C_{\a\b}>0$ such that 
\[
|\pa_{x_0}^\a\pa_{\x_0}^\b y(t)|\leq C_{\a\b} \jap{{x_0}}^{1-|\a|-\m}, 
\quad |\pa_{x_0}^\a \pa_{\x_0}^\b \x(t)|\leq C_{\a\b}\jap{{x_0}}^{-|\a|-\m},
\]
uniformly for $\pm t\geq 0$, $({x_0},{\x_0})\in \O_{I_5,\pm}(\b)$. 
\end{enumerate}
\end{lem}

\begin{proof}
We recall that estimates for $\x(t)$ have already proved in Lemma~\ref{lem-cl-ogic2}. 
As before, we consider the ``$+$''-case only. 

(i) By the Hamilton equation, we have 
\[
\pa_t(\pa_z y) = (\pa_y\pa_\x q)\pa_z y + (\pa_\x\pa_\x q)\pa_z \x, 
\]
where $z=x_0$ or $\x_0$. By Lemma~\ref{lem-cl-ipsym1} and Lemma~\ref{lem-cl-ogic2}, we learn 
\begin{align}
\pa_t|\pa_{x_0} y| &\leq 
C\jap{x_0;t}^{-1-\m}|\pa_{x_0}y| + C\jap{x_0}\jap{x_0;t}^{-1}\jap{t}^{-\m}\jap{x_0}^{-1-\m}
\label{eq-cl-c4}\\
&\leq C\jap{t}^{-1-\m} |\pa_{x_0}y| + C \jap{t}^{-1-\m} \jap{x_0}^{-\m}\nonumber
\end{align}
with $\pa_{x_0}y(0)=\mathrm{E}$. Then by Gronwall's inequality, 
we have $|\pa_{x_0}y(t)|\leq C$ uniformly in $t\geq 0$. 
Substituting this to \eqref{eq-cl-c4} again to learn 
\[
\pa_t|\pa_{x_0}y(t)|\leq C \jap{x_0;t}^{-1-\m}+C\jap{x_0}^{-\m}\jap{t}^{-1-\m},
\]
and this implies $|\pa_{x_0}y(t)-\mathrm{E}|\leq C \jap{x_0}^{-\m}$, 
thanks to Lemma~\ref{lem-cl-elementary}. 

Similarly, we have 
\begin{align*}
\pa_t|\pa_{\x_0} y| 
&\leq C \jap{x_0;t}^{-1-\m}|\pa_{\x_0}y|+ C\jap{x_0}\jap{x_0;t}^{-1}\jap{t}^{-\m}\\
&\leq C\jap{t}^{-1-\m}|\pa_{\x_0}y| + C \jap{x_0}\jap{x_0;t}^{-1}\jap{t}^{-\m}\nonumber 
\end{align*}
with $\pa_{x_0}y(0)=0$. 
Then by Gronwall's inequality and Lemma~\ref{lem-cl-elementary}, we learn
\[
|\pa_{\x_0}y(t)|\leq C \jap{x_0}^{1-\m}, \quad t\geq 0. 
\]
(ii) We prove the claim by induction in $k=|\a+\b|$. We use the induction hypothesis:  
\begin{quote}
${(D_k):}$ For $\a,\b\in\ze_+^d$, $|\a+\b|\leq k$, there is $C_{\a\b}>0$ 
such that 
\[
\bigabs{\pa_{x_0}^\a\pa_{\x_0}^\b y(t)}\leq C_{\a\b}\jap{x_0}^{1-|\a|},
\quad t\geq 0.
\]
uniformly for $(x_0,\x_0)\in \O_{I_5}$ such that $\cos(x_0,v(\x_0))\geq \b$. 
\end{quote}
We note $(D_0)$ is an immediate consequence of Lemma~\ref{lem-cl-iptr1}, 
and $(D_1)$ is proved in (i). We also recall 
\[
\bigabs{\pa_{x_0}^\a\pa_{\x_0}^\b \x(t)}\leq C'_{\a\b}\jap{x_0}^{-|\a|}
\]
with some constant $C_{\a\b}'>0$ for all $\a,\b\in \ze_+^d$ (including $\a=0$). 

Suppose $(D_{N-1})$ holds, and let $|\a+\b|=N$. 
As in the proof of Lemma~\ref{lem-cl-ogic2}, we use the Leibniz formula: 
\begin{align*}
\pa_t(\pa_{x_0}^\a\pa_{\x_0}^\b y) 
&= (\pa_\x\pa_y  q) \pa_{x_0}^\a\pa_{\x_0}^\b y + (\pa_\x\pa_\x q)\pa_{x_0}^\a\pa_{\x_0}^\b\x \\
& +\sum_* c_*(\pa_\x \pa_y^\c\pa_\x^\d q) \prod_{i=1}^{d} 
\biggpare{\prod_{j=1}^{\c_i}(\pa_{x_0}^{a(i,j)}\pa_{\x_0}^{b(i,j)} y_i)
\prod_{k=1}^{\d_i} (\pa_{x_0}^{\tilde a(i,k)}\pa_{\x_0}^{\tilde b(i,k)}\x_i)}, 
\end{align*}
where the last sum is taken over $\c,\d, a(i,j), \tilde a(i,j), b(i,k), \tilde b(i,k)\in\ze_+^d$ such that 
$|\c+\d|\geq 2$; $a(i,j)+b(i,j)\neq 0$ for $i=1,\dots,d$, $j=1,\dots,\c_i$;  
$\tilde a(i,k)+\tilde b(i,k)\neq 0$ for $i=1,\dots,d$, $k=1,\dots,\d_i$; and  
\[
\sum_{i=1}^d\biggpare{\sum_{j=1}^{\c_i}a(i,j) +\sum_{k=1}^{\d_i}\tilde a(i,k)} =\a, \quad 
\sum_{i=1}^d\biggpare{\sum_{j=1}^{\c_i}b(i,j) +\sum_{k=1}^{\d_i}\tilde b(i,k)} =\b.
\]
$c_*$ denotes suitable universal constant for each index. 
We denote the last term by $R_3$. 
By the induction hypothesis and Corollary~\ref{cor-cl-ipsym2}, we learn 
\[
|R_3|\leq C \jap{x_0}^{1-|\a|}\jap{x_0;t}^{-1}\jap{t}^{-\m}.
\]
Hence we have 
\begin{align*}
\pa_t\bigabs{\pa_{x_0}^\a\pa_{\x_0}^\b y}
&\leq C\jap{x_0;t}^{-1-\m}\bigabs{\pa_{x_0}^\a\pa_{\x_0}^\b y}
+C\jap{x_0} \jap{x_0;t}^{-1}\jap{t}^{-\m} \jap{x_0}^{-|\a|} \\
&\quad +C\jap{x_0}^{1-|\a|}\jap{x_0;t}^{-1}\jap{t}^{-\m}\\
&\leq C\jap{t}^{-1-\m}\bigabs{\pa_{x_0}^\a\pa_{\x_0}^\b y}
+C' \jap{x_0}^{1-|\a|}\jap{x_0;t}^{-1}\jap{t}^{-\m}
\end{align*}
with $\pa_{x_0}^\a\pa_{\x_0}^\b y(0)=0$. Thus by Gronwall's inequality and 
Lemma~\ref{lem-cl-elementary} again, we have 
\[
|\pa_{x_0}^\a\pa_{\x_0}^\b y(t)|\leq C\jap{x_0}^{1-|\a|-\m},
\]
which proves the induction step $(D_N)$, and also completes the proof. 
\end{proof}

\begin{cor}\label{cor-cl-iptr3}
Under the assumptions of Lemma~\ref{lem-cl-iptr1}, 
\[
|y(x_0,\x_0,t)-x_0|\leq C|x_0|^{1-\m}
\]
uniformly in $\pm t\geq 0$, $({x_0},{\x_0})\in \O_{I_5,\pm}(\b)$.
\end{cor}

\begin{proof}
We note $y(0,\x_0,t)=0$ for all $t\in\re$ by the definition. 
Hence, 
\begin{align*}
|y(x_0,\x_0,t)-x_0|
&=\biggabs{\int_0^1 \frac{d}{ds} (y(sx_0,\x_0,t)-sx_0)ds}\\
&= \biggabs{\int_0^1 x_0\cdot \bigbra{(\pa_{x_0}y)(sx_0,\x_0,t)-\mathrm{E}}ds}\\
&\leq C|x_0|\int_0^1 \jap{sx_0}^{-\m} ds
\leq C |x_0|\int_0^1 |sx_0|^{-\m}ds \\
&= C|x_0|^{1-\m}\int_0^1s^{-\m}ds =C'|x_0|^{1-\m}.
\end{align*}
\end{proof}


\subsection{Solutions to Hamilton-Jacobi equation in the interaction picture}
\label{subsection-cl-ipHJ}

In this subsection, we construct a solution to the Hamilton-Jacobi equation in the interaction picture: 
\[
\pa_t\g(t,x_0,\x)= q(t,\pa_\x\g(t,x_0,\x),\x), \quad \g(0,x_0,\x)=x_0\cdot\x,
\]
and study its properties. We write, 
\[
\L_t^{x_0}\ :\ \x_0\mapsto \x(t,x_0,\x_0)\in\re^d 
\]
for $(x_0,\x_0)\in \O_{I_5}$ and $t\in\re$. 
By Lemma~\ref{lem-cl-convexity}, $\L_t^{x_0}$ is a diffeomorphism from $\O_{I_4}^0$ 
into a subset of $\O_{I_5}^0$, which contains $\O_{I_3}^0$ 
for each $t\in\re$, $x_0\in\re^d$, 
and the inverse has uniformly bounded Jacobian matrix on $\O_{I_3}^0$. 

We set
\[
\f(t,x_0,\x_0) = \int_0^t \bigbra{q(s,y(s),\x(s))-y(s)\cdot\pa_y q(s,y(s),\x(s))} ds +x_0\cdot\x_0,
\]
where 
\[
y(t)=y(x_0,\x_0;t), \quad \x(t)=\x(x_0,\x_0;t), \quad t\in\re.
\]
Then by the standard theory of the Hamilton-Jacobi equation, 
\[
\g(t,x_0,\x)= \f(t,x_0,(\L_t^{x_0})^{-1}(\x))
\]
satisfies the above Hamilton-Jacobi equation and the initial condition. 
We also recall that $\g(t,x,\x)$ is the generating function of the evolution 
(see, e.g., \cite{Arnold} \S47), 
\[
w_t\ :\ (x_0,\x_0)\mapsto (y(t),\x(t)), 
\]
namely, we have 
\[
w_t\ :\ \begin{pmatrix} x \\ \pa_x\g(t,x,\x) \end{pmatrix} \mapsto 
\begin{pmatrix} \pa_\x\g(t,x,\x) \\ \x \end{pmatrix}.
\]
Then, the conservation of the energy is expressed as 
\begin{equation}\label{interaction-picture-energy-conservation}
p(x,\pa_x\g(t,x,\x)) = p(\pa_\x\g(t,x,\x)+\pa_\x\phi(t,\x),\x), 
\end{equation}
provided $(x_0,\pa_x\g(t,x,\x))\in \O_{I_4}$. 

We show $\g(t,x,\x)$ is a good symbol on $\O_{I_3,\pm}(\b)$ 
for $\pm t\geq 0$, respectively. 

\begin{lem}\label{lem-cl-ipHJ1}
For $\a,\b\in\ze_+^d$, there is $C_{\a\b}>0$ such that 
\[
\bigabs{\pa_{x_0}^\a\pa_\x^\b \bigpare{(\L_t^{x_0})^{-1}(\x)}}\leq C_{\a\b}\jap{x_0}^{-|\a|},
\]
uniformly in $x_0,\x$, $\pm t\geq 0$, provided $\x=\L_t^{x_0}(\x_0)$ with $(x_0,\x_0)\in \O_{I_5,\pm}(\b)$. Moreover, 
\[
\bigabs{\pa_{x_0}^\a\pa_\x^\b \bigpare{(\L_t^{x_0})^{-1}(\x)-\x}}\leq C_{\a\b}\jap{x_0}^{-|\a|-\m}
\]
under the same conditions. 
\end{lem}

\begin{proof}
By Lemma~\ref{lem-cl-ogic2} and the definition of $\L_t^{x_0}$, we learn
\[
\pa_\x^\a \pa_{\x_0}^\b (\L_t^{x_0}(\x)) = O(\jap{x_0}^{-|\a|}),
\]
and
\[
\pa_{x_0}^\a \pa_{\x}^\b (\L_t^{x_0}(\x)-\x) = O(\jap{x_0}^{-|\a|-\m}).
\]
(Note we did not prove the estimate in Lemma~\ref{lem-cl-ogic2} for $\a=\b=0$, but this is easily shown as well.)
By Lemma~\ref{lem-cl-convexity}, $(\pa \x/\pa \x_0)^{-1}$ is uniformly bounded, and hence 
by the standard formulas of the derivatives of inverse map, we obtain the claim. 
\end{proof}

In particular, we learn 
\[
\bigabs{\pa_{x_0}^\a\pa_\x^\b \bigpare{x_0\cdot (\L_t^{x_0})^{-1}(\x)-x_0\cdot\x)}}
\leq C'_{\a\b}\jap{x_0}^{1-\m-|\a|}.
\]

\begin{lem}\label{lem-cl-ipHJ2}
For $\a,\b\in\ze_+^d$, there is $C_{\a\b}>0$ such that 
\[
\bigabs{\pa_{x_0}^\a\pa_{\x_0}^\b \bigpare{\f(t,x_0,\x_0)-x_0\cdot\x_0}}\leq C_{\a\b} \jap{x_0}^{1-\m-|\a|}, 
\]
uniformly for $(x_0,\x_0)\in\O_{I_5,\pm}(\b)$, $\pm t\geq 0$, respectively.
\end{lem}

\begin{proof}
We write
\[
\ell(t,y,\x)= q(t,y,\x)-y\cdot\pa_y q(t,y,\x).
\]
Then by Lemma~\ref{lem-cl-ipsym1} (or by Corollary~\ref{cor-cl-ipsym2}), we learn, for any $\a,\b\in\ze_+^d$, 
\[
\bigabs{(\pa_{y}^\a\pa_{\x}^\b \ell)(t,y(t),\x(t))}\leq C\jap{x_0}^{1-|\a|}\jap{x_0;t}^{-1}\jap{t}^{-\m}.
\]
Combining these with Lemma~\ref{lem-cl-iptr2}, we have, for any $\a,\b\in\ze_+^d$, 
\[
\bigabs{\pa_{x_0}^\a\pa_{\x_0}^\b (\ell(t,y(t),\x(t)))}\leq C\jap{x_0}^{1-|\a|}\jap{x_0;t}^{-1}\jap{t}^{-\m}
\]
with some $C>0$. Integrating this in $t$, and using Lemma~\ref{lem-cl-elementary}, we obtain the claim. 
\end{proof}

Combining these two lemmas, we obtain the following estimate: 

\begin{lem}\label{lem-cl-ipHJ3}
For $\a,\b\in\ze_+^d$, there is $C_{\a\b}>0$ such that 
\[
\bigabs{\pa_{x}^\a\pa_{\x}^\b \bigpare{\g(t,x,\x)-x\cdot\x}}\leq C_{\a\b} \jap{x}^{1-\m-|\a|}, 
\]
uniformly in $x,\x$, $\pm t\geq 0$, provided $(x,\x_0)\in \O_{I_5,\pm}(\b)$ 
with $\x=\L_t^{x}(\x_0)$. 
\end{lem}

We address the conditions on the domain in the above lemma later. 
We here note that the condition is satisfied if $(x,\x)\in \O_{I_3,\pm}(\b')$ 
with $\b'>\b$ and $|x|\gg 0$.

\begin{rem}
In the above results, we concentrate on the properties of functions on 
$\O_{I_5,\pm}(\b)$. These functions are defined globally (provided $(x,\x)\in \O_{I_5}$), 
and they are smooth. The same analysis can be easily carried out locally in $(x,\x)$
in a neighborhood of any arbitrarily fixed point in $\O_{I_2}$. 
Actually, by Lemma~\ref{lem-cl-convexity}, we can show $\pa_x^\a\pa_\x^\b\g(t;x,\x)$ is uniformly 
bounded by $O(\jap{x})$, but it does not satisfy the above properties globally. 
\end{rem}

\subsection{Classical wave maps and their generating functions}
\label{subsection-cl-wm}

We have already seen in Subsection~\ref{subsection-cl-ip} that $\lim_{t\to\pm\infty} (y(t),\x(t))$
exist. We denote them by 
\[
w_\pm : (x_0,\x_0)\mapsto (x_\pm,\x_\pm)=\lim_{t\to\pm\infty}(y(t),\x(t))
\]
and we call them {\it classical (inverse) wave maps}. 
By the results (and the proof) in Subsection~\ref{subsection-cl-ip}, we can easily show: 

\begin{lem}\label{lem-cl-wm1}
$x_\pm$, $\x_\pm$ are smooth functions of $(x_0,\x_0)\in \O_{I_5}$, 
and for any $\a,\b\in\ze_+^d$, 
\[
\lim_{t\to\pm\infty} \pa_{x_0}^\a\pa_{\x_0}^\b y(t) 
=\pa_{x_0}^\a\pa_{\x_0}^\b x_\pm, \quad
\lim_{t\to\pm\infty} \pa_{x_0}^\a\pa_{\x_0}^\b \x(t) 
=\pa_{x_0}^\a\pa_{\x_0}^\b \x_\pm. 
\]
Moreover, for $(x_0,\x_0)\in \O_{I_5,\pm}(\b)$, 
\begin{align*}
&\bigabs{\pa_{x_0} x_\pm-\mathrm{E}}\leq C\jap{x_0}^{-\m}, \quad \bigabs{\pa_{\x_0}x_\pm}\leq C\jap{x_0}^{-1-\m}, \\
&\bigabs{\pa_{x_0} \x_\pm}\leq C\jap{x_0}^{-1-\m}, 
\quad \bigabs{\pa_{\x_0}\x_\pm-\mathrm{E}}\leq C\jap{x_0}^{-\m},
\end{align*}
and 
\[
\bigabs{\pa_{x_0}^\a\pa_{\x_0}^\b x_\pm}\leq C_{\a\b}\jap{x_0}^{1-|\a|-\m}, \quad 
\bigabs{\pa_{x_0}^\a\pa_{\x_0}^\b \x_\pm}\leq C_{\a\b}\jap{x_0}^{-|\a|-\m}
\]
if $|\a+\b|\geq 2$. 
The convergence is uniform in the following sense: 
\begin{align*}
&\lim_{t\to\pm\infty} \sup_{(x_0,\x_0)\in\O_{I_2,\pm}(\b)}\jap{x_0}^{-1+|\a|} 
\bigabs{\pa_{x_0}^\a\pa_{\x_0}^\b y(t) -\pa_{x_0}^\a\pa_{\x_0}^\b x_\pm}=0, \\
&\lim_{t\to\pm\infty} \sup_{(x_0,\x_0)\in\O_{I_2,\pm}(\b)}\jap{x_0}^{|\a|} 
\bigabs{\pa_{x_0}^\a\pa_{\x_0}^\b \x(t) 
-\pa_{x_0}^\a\pa_{\x_0}^\b \x_\pm}=0, 
\end{align*}
\end{lem}

We also have the limit of the generating function of $w_t$: 
\[
\g_\pm(x,\x) =\lim_{t\to\pm\infty} \g(t,x,\x).
\]

\begin{lem}\label{lem-cl-wm2}
$\g_\pm(x,\x)$ is smooth functions of $(x,\x)\in \O_{I_4}$ and and for any $\a,\b\in\ze_+^d$,
\[
\lim_{t\to\pm\infty} \pa_x^\a\pa_\x^\b \g(t,x,\x)=\pa_x^\a\pa_\x^\b \g_\pm(x,\x).
\]
Moreover, 
\[
\bigabs{\pa_x^\a\pa_\x^\b\bigpare{\g_\pm(x,\x)-x\cdot\x}}\leq C_{\a\b}\jap{x}^{1-\m-|\a|}, 
\quad (x,\x)\in\tilde\O_{I_5,\pm}(\b),
\]
where 
\[
\tilde\O_{J,\pm}(\b)=\bigset{(x,\x)}{(x,\pa_x\g_\pm(x,\x))\in\O_{J,\pm}(\b)}.
\]
The convergence is uniform in the sense: 
\[
\lim_{t\to\pm\infty} \sup_{(x,\x)\in\tilde\O_{I_5,\pm}(\b)}\jap{x}^{-1+\m+|\a|} 
\bigabs{\pa_{x}^\a\pa_{\x}^\b \g(t,x,\x) -\pa_{x}^\a\pa_{\x}^\b \g_\pm(x,\x)}=0.
\]
\end{lem}

We also note that $\g_\pm$ is the generating function of $w_\pm$, i.e., 
\[
w_\pm\ : \ 
\begin{pmatrix}x\\ \pa_x\g_\pm(x,\x)\end{pmatrix}
\mapsto
\begin{pmatrix}\pa_\x\g_\pm(x,\x)\\ \x \end{pmatrix}. 
\]
The energy conservation is 
\begin{align*}
p(x,\pa_\x\g_\pm(x,\x)) &=\lim_{t\to\pm\infty} p(y(t)+\pa_\x\phi(t,\x),\x)\\
&=\lim_{t\to\pm\infty} \bigpare{p_0(\x)+V_R(y(t)+\pa_\x\phi(t,\x),\x)}, 
\end{align*}
but since $y(t)$ is uniformly bounded and $|\pa_\x\phi(t,\x)|\to\infty$, we learn 
that the right hand side converges to $p_0(\x)$. Thus $\g_\pm(x,\x)$ are solutions to 
the eikonal equation: 
\begin{equation}\label{eq-eikonal}
p(x,\pa_x\g_\pm(x,\x))=p_0(\x), \quad (x,\x)\in\tilde\O_{I_5,\pm}(\b). 
\end{equation}

Finally, we consider the definition domain of the generating function $\g_\pm(x,\x)$, 
i.e., $\tilde\O_{I_5,\pm}(\b)$. 

\begin{lem}\label{lem-cl-wm3}
Let $\b'>\b$. Then there is $L>0$ such that 
\[
\O_{I_4,\pm}^0(\b',L)\subset \tilde\O_{I_5,\pm}(\b),
\]
where
\[
\O_{J,\pm}^0(\c,L)=\bigset{(x,\x)}{p_0(\x)\in J, |x|\geq L, \pm \cos(x,v(\x))>\c}
\]
for $J\subset\re$, $L>0$ and $\c>-1$. 
\end{lem}

\begin{proof}
For a given $(x,\x)\in \O^0_{I_4,\pm}(\b',L)$, it suffices to find 
$(x,\x_0)\in \tilde\O_{I_5,\pm}(\b)$ such that $\x=\x_\pm(x,\x_0)$. 
Thus we find a inverse map of $\x_0\mapsto \x=\x_\pm(x,\x_0)$ for such $(x,\x)$. 
We construct the inverse map by, 
for example, the contraction mapping. For a fixed $(x,\x)$, we set 
\[
F_\pm(\y)= \x-(\x_\pm(x,\y)-\y), \quad \y\in \O_{I_5},
\]
and then $F_\pm(\y)=\y$ if and only if $\x= \x_\pm(x,\x)$. 
We note, by the construction of $\x_\pm$ and Lemma~\ref{lem-cl-wm1}, we have 
\[
|\x_\pm(x,\y)-\y|\leq C\jap{x}^{-\m}, \quad 
|\pa_\y\x_\pm(x,\y) -\mathrm{E}|\leq C\jap{x}^{-\m}, 
\]
uniformly for $(x,\y)\in\O_{I_5,\pm}(\b')$. Thus, if $|x|\geq L$ is sufficiently large, 
$F_\pm$ is a contraction map in a small ball with the center at $\x_0$ which is 
contained in $\O_{I_5}^0$ and $\{\y\,|\,\pm\cos(x,\y)>\b\}$. 
Thus we can apply the fixed point theorem to conclude the existence of the fixed point. 
This implies the assertion. 
\end{proof}

The above lemma implies we can apply the result of Lemma~\ref{lem-cl-wm2} for 
$(x,\x)\in \O_{I_4,\pm}^0$. We note that the phase function $\g_\pm$ is well-defined 
on $\re^d\times \O_{I_4}^0$, though they do not enjoy the decay properties in 
Lemma~\ref{lem-cl-wm2} globally:

\begin{lem}\label{lem-cl-wm4}
$\g_\pm(x,\x)$ is well-defined for $x\in\re^d$, $\x\in \O_{I_5}^0$, namely, 
\[
\re^d\times \O_{I_4}^0 \subset \tilde \O_{I_5,\pm}= \bigset{(x,\x)}{(x,\pa\g_\pm(x,\x))\in \O_{I_5}}. 
\]
\end{lem}

\begin{proof}
We consider the ``$+$'' case only. 
It suffices to show that for $x\in\re^d$, $\x\in\O_{I_4}^0$ there is $\x_0$ 
such that $\x_+(x,\x_0)=\x$. We recall, by Lemma~\ref{lem-cl-convexity} and the condition on $V_R$, 
for each $x\in\re^d$ the map
\[
\x_0\mapsto \x_+(x_0,\x_0)= \L^{x_0}_+(x_0) =\lim_{t\to\infty} \L^{x_0}_t(\x_0)
\]
is diffeomorphism from $\{\x\,|\, p(x_0,\x)\in I_5\}$ into $\O^0_{I_6}$, and the range covers
$\O^0_{I_4}$. Hence the inverse map is well-defined on $\O^0_{I_4}$, and 
hence $x_0=(\L_+^{x_0})^{-1}(\x)$ satisfies the required property. 
We note that by the eikonal equation \eqref{eq-eikonal}, we learn $p(x_0,\x_0)=p_0(\x)\in I_4$, 
and hence $(x_0,\x_0)\in \O_{I_4}$. 
\end{proof}


\section{Time-independent modifiers}\label{section-isozaki-kitada}

We construct the so-called Isozaki-Kitada modifiers, or time-independent modifiers, 
$J_\pm$, using solutions of eikonal equations $\g_\pm(x,\x)$ constructed in Section~\ref{section-cl}. 
We suppose $J_\pm$ has the form 
\[
J_\pm f(x) =(2\pi)^{-d/2}\int_{\re^d} e^{i\g_\pm(x,\x)} b_\pm(x,\x)\hat f(\x)d\x
\]
for $f\in\mathcal{S}(\re^d)$, and the symbols $b_\pm(x,\x)$ are elements of 
$S(1,dx^2/\jap{x}^2+d\x^2)$, and supported in $\tilde\O_{I_4,\pm}(\b)$ with some 
$-1<\b\leq 1$. 
Our construction is analogous to the one in Derezi\'nski-G\'erard \cite{Derezinski-Gerard} \S 4.15 
and Robert \cite{Robert}, though the setting is more general. 
We construct $b_\pm$ in the rest of this section. 
We mostly consider the ``$+$''-case. The other case can be handled similarly. 

We suppose $a^\pm$ has the form 
\[
b_\pm(x,\x)= \Th_\pm(1+a_1^\pm+a_2^\pm+\cdots), 
\]
where $\Th_\pm\in S(1,g)$, $a_j^\pm\in S(\jap{x}^{-\m-j},g)$, $j=1,2,\dots$, 
on $\tilde\O_{I_4,\pm}(\b)$, where we denote $g=dx^2/\jap{x}^2+d\x^2$. 
We construct these symbols so that 
\[
H J_\pm -J_\pm H_0\sim 0
\]
asymptotically as $|x|\to\infty$ in $\tilde\O_{I_4,\pm}(\b)$. 
At first we prepare a formula to compute $HJ_\pm$: 

\begin{lem}\label{lem-ik-osc-formula}
Suppose $b_\pm \in S(\jap{x}^{\n},g)$, $\n\in\re$, and  supported in $\tilde\O_{I_4,\pm}(\b)$ 
with some $\b>-1$. Then 
\begin{align*}
& e^{-i\g_\pm(x,\x)} H \bigbrac{e^{i\g_\pm(\cdot,\x)}b_\pm(\cdot,\x)} 
=p(x,\pa_x\g_\pm(x,\x))b_\pm(x,\x) \\
&\quad -\frac{i}{2}\pa_x\cdot \bigpare{(\pa_\x p)(x,\pa_x\g_\pm(x,\x))}b_\pm(x,\x) 
-i(\pa_\x p)(x,\pa_x \g_\pm(x,\x))\cdot\pa_x b_\pm(x,\x)\\
&\quad +r_\pm(x,\x),
\end{align*}
with $r_\pm\in S(\jap{x}^{-2+\n-\m},g)$. Moreover, $r_\pm$ are supported essentially in 
$\tilde\O_{I_4,\pm}(\b)$, i.e., they decay rapidly in $x$, away from $\tilde\O_{I_4,\pm}(\b)$. 
\end{lem}

\begin{proof}
We compute 
\begin{align*}
&e^{-i\g_\pm(x,\x)} H [e^{i\g_\pm(\cdot,\x)}b_\pm(\cdot,\x)]
 = e^{-i\g_\pm(x,\x)}p^W\!(x,D_x)[ e^{i\g_\pm(\cdot,\x)}b_\pm(\cdot,\x)]\\
&\quad = (2\pi)^{-d}\iint e^{-i\g_\pm(x,\x)+i(x-y)\cdot\y+i\g_\pm(y,\x)}p(\tfrac{x+y}{2},\y)
b_\pm(y,\x)dyd\y \\
&\quad = (2\pi)^{-d} \iint e^{i(x-y)\cdot(\y-\int_0^1 \pa_x\g_\pm(tx+(1-t)y,\x)dt)}
p(\tfrac{x+y}{2},\y)b_\pm(y,\x)dyd\y \\
&\quad = (2\pi)^{-d}\iint e^{i(x-y)\cdot\y} p(\tfrac{x+y}{2},\y+\F_\pm(x,y,\x))
b_\pm(y,\x)dyd\y,
\end{align*}
where 
\begin{align*}
\F_\pm(x,y,\x)&= \int_0^1 \pa_x \g_\pm(tx+(1-t)y,\x)dt\\
&= \int_{-1/2}^{1/2} \pa_x\g_\pm(\tfrac{x+y}{2}+t(x-y),\x)dt.
\end{align*}
We note $\F_\pm(x,y,\x)$ are even functions in $x-y$, and hence $(\pa_x-\pa_y)\F_\pm(x,y,\x)=0$. 
Moreover, we have 
\[
\F_\pm(x,y,\x)-\x = \int_{-1/2}^{1/2}(\pa_x\g_\pm(\tfrac{x+y}{2}+t(x-y),\x)-\x)dt
\]
and hence, by Lemma~\ref{lem-cl-ipHJ3}, 
\begin{equation}\label{eq-ik-phase}
\bigabs{\pa_x^\a\pa_y^\b\pa_\x^\c (\F_\pm(x,y,\x)-\x)}\leq C_{\a\b\c}\jap{x+y}^{-\m-|\a|-|\b|}, 
\quad\text{if } \biggabs{\frac{x}{|x|}-\frac{y}{|y|}}<\d\ll 1.
\end{equation}
We now use Taylor expansion in $\y$ to learn 
\begin{align*}
p(\tfrac{x+y}{2},\y+\F(x,y,\x))
&=p(\tfrac{x+y}{2},\F(x,y,\x)) + \y\cdot(\pa_\x p)(\tfrac{x+y}{2},\F(x,y,\x))\\
&\quad +\frac12 \int_0^1 \sum_{j,k}\y_j\y_k (\pa_{\x_j}\pa_{\x_k}p)(\tfrac{x+y}{2},t\y+\F(x,y,\x))dt.
\end{align*}
We substitute this to the above equation to obtain
\[
e^{-i\g_\pm(x,\x)} H [e^{i\g_\pm(\cdot,\x)}b_\pm(\cdot,\x)]=I_1+I_2+I_3, 
\]
where 
\begin{align*}
\mathrm{I}_1 &=(2\pi)^{-d}\iint e^{i(x-y)\cdot\y} p(\tfrac{x+y}{2},\F(x,y,\x))b_\pm(y,\x)dyd\y, \\
\mathrm{I}_2&= (2\pi)^{-d}\iint e^{i(x-y)\cdot\y} \y\cdot(\pa_\x p)(\tfrac{x+y}{2},\F(x,y,\x))b_\pm(y,\x)
dyd\y, \\
\mathrm{I}_3&=(2\pi)^{-d}\iint\int_0^1 \sum_{j,k} e^{i(x-y)\cdot\y} \y_j\y_k (\pa_{\x_j}\pa_{\x_k}p)
(\tfrac{x+y}{2},t\y+\F(x,y,\x))\times \\
&\qquad \times b_\pm(y,\x)dtdyd\y.
\end{align*}
By oscillatory integrations, we have 
\begin{align*}
\mathrm{I}_1 &= p(x,\pa_x\g_\pm(x,\x))b_\pm(x,\x), \\
\mathrm{I}_2&= -(2\pi)^{-d}\iint e^{i(x-y)\cdot\y} i\pa_y\cdot\bigbra{(\pa_\x p)(\tfrac{x+y}{2},\F(x,y,\x))
b_\pm(y,\x)}dyd\y \\
& =-\frac{i}{2}\pa_x\cdot \bigbra{(\pa_\x p)(x,\pa_x\g_\pm(x,\x))}b_\pm(x,\x)
-i (\pa_\x p)(x,\pa_x\g_\pm(x,\x))\cdot  \pa_x b_\pm(x,\x). 
\end{align*}
By virtue of \eqref{eq-ik-phase}, and using integration by parts, 
we also have $r_\pm=\mathrm{I}_3 \in S(\jap{x}^{-2+\n-\m},g)$.
It is easy to observe $r_\pm$ are essentially supported in $\tilde\O_{I_2,\pm}(\b)$. 
\end{proof}

We now compute the 0-th order term $\Th_\pm(x,\x)$ in the above setting. 
This factor is actually the well-known volume factor in the WKB analysis. 

\begin{lem} \label{lem-ik-Th-solution}
Let $\Th_\pm(x,\x)= \Bigpare{\det\Bigpare{\frac{\pa^2\g_\pm}{\pa x\pa\x}}}^{1/2}$, 
then $\Th_\pm$ satisfies 
\[
\frac12 \pa_x\cdot \bigbra{(\pa_\x p)(x,\pa_x\g_\pm(x,\x))}\Th_\pm(x,\x)
+(\pa_\x p)(x,\pa_x\g_\pm(x,\x))\cdot \pa_x \Th_\pm(x,\x) =0. 
\]
Moreover, $\Th_\pm-1 \in S(\jap{x}^{-\m},g)$ on $\tilde\O_{I_4,\pm}(\b)$. 
\end{lem}

\begin{proof}
By differentiating the eikonal equation \eqref{eq-eikonal} in $\x_j$, we learn 
\[
\sum_{k=1}^d \pa_{\x_j}\pa_{x_k}\g_\pm(x,\x)\pa_{\x_k}p(x,\pa_x\g_\pm(x,\x))-\pa_{\x_j}p_0(\x) =0.
\]
Then we differentiate this in $x_i$: 
\[
\sum_{k=1}^d \pa_{\x_k}p \, \pa_{x_k}(\pa_{x_i}\pa_{\x_j}\g_\pm)
+\sum_{k=1}^d \pa_{\x_j}\pa_{x_k}\g_\pm \, \pa_{x_i}(\pa_{\x_k}p(x,\pa_x\g_\pm))=0.
\]
We write this in matrix form to obtain
\[
\sum_{k=1}^d (\pa_{\x_k}p) \frac{\pa}{\pa x_k}\biggpare{\frac{\pa^2\g_\pm}{\pa x\pa \x}}
+\biggbrac{\frac{\pa}{\pa x}\biggpare{\frac{\pa p}{\pa \x}(x,\pa_x\g_\pm)}}
\biggpare{\frac{\pa^2 \g_\pm}{\pa x\pa \x}}=0. 
\]
Since $\pa_x\pa_\x \g_\pm$ is invertible (as a matrix), we have 
\[
\sum_{k=1}^d (\pa_{\x_k}p) \biggbrac{\frac{\pa}{\pa x_k}\biggpare{\frac{\pa^2\g_\pm}{\pa x\pa \x}}}
\biggpare{\frac{\pa^2 \g_\pm}{\pa x\pa \x}}^{-1}
+\biggbrac{\frac{\pa}{\pa x}\biggpare{\frac{\pa p}{\pa \x}(x,\pa_x\g_\pm)}}
=0. 
\]
Then we take the trace: 
\[
\sum_{k=1}^d (\pa_{\x_k}p) \trace\biggbrac{\biggbrac{\frac{\pa}{\pa x_k}\biggpare{\frac{\pa^2\g_\pm}{\pa x\pa \x}}}
\biggpare{\frac{\pa^2 \g_\pm}{\pa x\pa \x}}^{-1}}
+\sum_{k=1}^d\frac{\pa}{\pa x_k}\biggpare{\frac{\pa p}{\pa \x_k}(x,\pa_x\g_\pm)}
=0. 
\]
On the other hand, by the derivative formula for the determinant, we learn 
\begin{align*}
\frac{\pa }{\pa x_k}\Th_\pm(x,\x)  
&= \frac12\frac{\pa}{\pa x_k}\biggbrac{\det\biggpare{\frac{\pa^2\g_\pm}{\pa x\pa \x}}}
\biggpare{\det\biggpare{\frac{\pa^2\g_\pm}{\pa x\pa \x}}}^{-1/2}\\
&=\frac12 \trace\biggbrac{\biggbrac{\frac{\pa}{\pa x_k}\biggpare{\frac{\pa^2\g_\pm}{\pa x\pa \x}}}
\biggpare{\frac{\pa^2 \g_\pm}{\pa x\pa \x}}^{-1}}\Th_\pm(x,\x), 
\end{align*}
and hence $\Th_\pm$ satisfies 
\[
\pa_\x p(x,\pa_x\g_\pm)\cdot \pa_x \Th_\pm
+\frac12 \pa_x\cdot((\pa_\x p)(x,\pa_x\g_\pm))\Th_\pm  =0.
\]
The last claim follows from the observation: $(\pa_x\pa_\x\g_\pm -E)\in S(\jap{x}^{-\m},g)$ 
on $\O_{I_4,\pm}(\b)$. 
\end{proof}

Now we construct symbols of $J_\pm$. By Lemmas~\ref{lem-ik-osc-formula}, 
\ref{lem-ik-Th-solution}, we learn, at least formally, 
\begin{align*}
&e^{-i\g_\pm(x,\x)} H [ e^{i\g_\pm(\cdot,\x)}\Th_\pm(\cdot,\x)] 
= p_0(\x)\Th_\pm(x,\x) \\
&\quad  -\frac{i}{2}\pa_x\cdot((\pa_\x p)(x,\pa_\x\g_\pm(x,\x)))\Th_\pm(x,\x) 
- i(\pa_\x p)(x,\pa_x\g_\pm(x,\x))\cdot \pa_x\Th_\pm(x,\x) \\
&\quad +r_\pm(x,\x) \\
&= p_0(\x)\Th_\pm(x,\x) + r_{0,\pm}(x,\x)\Th_\pm(x,\x)
\end{align*}
where $r_{0,\pm}\in S(\jap{x}^{-2-\m},g)$ on $\O_{I_4,\pm}(\b)$  with any $\b>-1$. 
We note $\Th_\pm$ do not satisfy the support property of Lemma~\ref{lem-ik-osc-formula}, 
but we will introduce cutoff functions, and the following computations are readily justified. 
By the same computation, for $b_\pm\in S(\jap{x}^{-\n},g)$, we have 
\begin{align*}
&e^{-i\g_\pm(x,\x)} H [ e^{i\g_\pm(\cdot,\x)}\Th_\pm(\cdot,\x)b_\pm(\cdot,\x)] 
= p_0(\x)\Th_\pm(x,\x) b_\pm(x,\x)\\
&\quad  
- i(\pa_\x p)(x,\pa_x\g_\pm(x,\x))\cdot (\pa_x b_\pm(x,\x))\Th_\pm(x,\x) 
+\tilde r_{0,\pm}(x,\x)\Th_\pm(x,\x)
\end{align*}
where $\tilde r_{0,\pm}\in S(\jap{x}^{-2-\m-\n},g)$ on $\O_{I_4,\pm}(\b)$. 
Thus, if we set
\[
a_1^\pm(x,\x) =i \int_0^{\pm\infty} r_{0,\pm}(\exp tH_p(x,\pa_x\g_\pm(x,\x))) dt, 
\]
then they solve the equations: 
\[
 (\pa_\x p)(x,\pa_x\g_\pm(x,\x))\cdot  \pa_x a_{1}^\pm(x,\x)= -i r_{0,\pm}(x,\x), 
\]
and hence
\begin{align*}
&e^{-i\g_\pm(x,\x)} H [ e^{i\g_\pm(\cdot,\x)}\Th_\pm(\cdot,\x)(1+a_1^\pm(\cdot,\x))] \\
&\quad = p_0(\x)\Th_\pm(x,\x)(1+ a_1^\pm(x,\x)) +r_{1,\pm}(x,\x)\Th_\pm(x,\x), 
\end{align*}
where $r_{1,\pm}\in S(\jap{x}^{-3-\m},g)$ on $O_{I_4,\pm}(\b)$. 
Moreover, $a_i^\pm$ satisfy the boundary condition: $a_1^\pm(x,\x)\to 0$ as $|x|\to\infty$ 
in $\O_{I_4,\pm}(\b)$. We note that if we set
\[
(z(t),\z(t))= \exp tH_p (x,\pa_x \g_\pm(x,\x)), 
\]
then $(z(t),\z(t))$ is the solution to the Hamilton equation with the 
boundary conditions $\z(t)\to\x$ as $t\to\pm\infty$, and $z(0)=x$.
Thus, by using Lemmas~\ref{lem-cl-ipHJ3} and \ref{lem-cl-ogic2}, we can show 
$a_1^\pm \in S(\jap{x}^{-1-\m},g)$ on $\O_{I_4,\pm}(\b)$.

We iterate this procedure to construct $a_j^\pm(x,\x)$, $j=2,3,\dots$. Namely, 
we set $r_{k,\pm}$ so that 
\begin{align*}
&r_{k,\pm}(x,\x)\Th_\pm(x,\x)\\
&\quad =e^{-i\g_\pm} H [e^{i\g_\pm}\Th_\pm(1+a_1^\pm+\cdots +a_k^\pm)] - 
p_0(\x)\Th_\pm(1+a_1^\pm+\cdots +a_k^\pm)\\
&\quad \in S(\jap{x}^{-k-2-\m},g)\text{ on }\O_{I_4,\pm}(\b).
\end{align*}
Then we solve the equation
\[
 (\pa_\x p)(x,\pa_x\g_\pm(x,\x))\cdot  \pa_x a_{k+1}^\pm(x,\x)= -i r_{k,\pm}(x,\x), 
\]
with the boundary condition: $a_k^\pm(x,\x)\to 0$ as $|x|\to\infty$ in $\O_{I_4,\pm}(\b)$. 
The solutions are given by 
\[
a_{k+1}^\pm(x,\x) =i \int_0^{\pm\infty} r_k^\pm(\exp tH_p(x,\pa_x\g_\pm(x,\x))) dt, 
\]
and we can show $a_{k+1}^\pm\in S(\jap{x}^{-k-1-\m},g)$ on $\O_{I_4,\pm}(\b)$ with any $\b>-1$. 

We define $a^\pm(x,\x)$ as an asymptotic sum of $1+a_1^\pm+\cdots$, i.e., 
$a^\pm\in S(1,g)$ on $\O_{I_4,\pm}(\b)$ such that for any $N\geq 1$, 
\[
a^\pm(x,\x)- \biggpare{1+\sum_{j=1}^N a_j^\pm(x,\x)}\in S(\jap{x}^{-N-2-\m},g)
\text{ on }\O_{I_4,\pm}(\b),
\]
with arbitrary $\b>-1$. 

Then we introduce a cut-off to these symbols. Let $R_0\gg 0$ and 
$-1<\b_{\pm,1}<\b_{\pm,2}< 1$. We choose smooth functions $\i_1(x)$, 
$\i_2(\l)$ and $\i_{3,\pm}(\s)$ such that 
\begin{align*}
&\i_1(x) =\i_1(|x|) =\begin{cases} 0 \quad &\text{if }|x|\leq 1, \\
1\quad & \text{if }|x|\leq 2, \end{cases}\\
&\i_2(\l) = \begin{cases} 1 \quad&\text{if }\l\in I_3, \\
0 \quad &\text{if }\l\notin I_4, \end{cases}\\
& \i_{3,\pm}(\s) =\begin{cases} 0 \quad &\text{if }\s\leq\b_{\pm,1},\\
1 \quad &\text{if }\s\geq \b_{\pm,2}, \end{cases}
\end{align*}
and $0\leq \i_1(x),\i_2(\l),\i_{3,\pm}(\s)\leq 1$. We then set 
\[
\i_\pm(x,\x) 
=\i_1(x/R_0)\i_2(p_0(\x))\i_{3,\pm}(\pm\cos(x,v(\pa_x\g_\pm(x,\x)))).
\]
We can now define our time-independent modifiers by 
\[
J_\pm f(x) =(2\pi)^{-d/2} \int e^{i\g_\pm(x,\x)}
\Theta_\pm(x,\x)
\i_\pm(x,\x) a^\pm(x,\x) \hat f(\x)d\x
\]
for $f\in \mathcal{S}(\re^d)$. 
On the support of the cut-off functions $\i_\pm(x,\x)$, the above formal computations 
can be readily justified, and we can show the following properties of $J_\pm$. 
We define interaction operators $G_\pm$ by 
\[
G_\pm = H J_\pm -J_\pm H_0, 
\]
which are bounded operators on $L^2(\re^d)$. 

\begin{lem} \label{lem-ik-G-formula}
There are symbols $g_\pm(x,\x)\in S(\jap{x}^{-1},g)$ such that 
\[
G_\pm f(x) = (2\pi)^{-d/2} \int e^{i\g_\pm(x,\x)}\Theta_\pm(x,\x)
g_\pm(x,\x) \hat f(\x)d\x
\]
for $f\in \mathcal{S}(\re^d)$. Moreover, $g_\pm$ are essentially supported in 
$\tilde\O_{I_4,\pm}(\b_{\pm,1})\setminus \tilde\O_{I_3,\pm}(\b_{\pm,2})$, i.e., 
for any $\a,\b\in\ze_+^d$ and $N$, there is $C_{\a\b N}>0$ such that 
\[
\bigabs{\pa_x^\a\pa_\x^\b g_\pm(x,\x)}
\leq C_{\a\b N}\jap{x}^{-N}, \quad 
(x,\x)\notin \tilde\O_{I_4,\pm}(\b_{\pm,1})\setminus \tilde\O_{I_3,\pm}(\b_{\pm,2}).
\]
The principal symbols of $g_\pm(x,\x)$ are given by 
$-i (\pa_\x p)(x,\pa_\x\g_\pm(x,\x))\cdot \pa_x\i_\pm(x,\x)$, i.e, 
\[
g_\pm(x,\x)-\bigbrac{-i (\pa_\x p)(x,\pa_x\g_\pm(x,\x))\cdot \pa_x\i_\pm(x,\x)}
\in S(\jap{x}^{-2},g).
\]
\end{lem}


\section{Wave operators, scattering operators, and scattering matrix}
\label{section-qm}

We follows the argument of \cite{N2016}, and we mainly explain the necessary modifications.
In the construction of $J_\pm$ in the last section, we choose $\b_{\pm,i}$, $i=1,2$, 
such that 
\[
-1<\b_{+,1}=\b_{-,1} <\b_{+,2}=\b_{-,2}<0,
\]
and fix them. We denote $\b_i=\b_{\pm,i}$, $i=1,2$. 
Using these modifiers $J_\pm$, we can now define 
wave operators with time-independent modifiers (or Isozaki-Kitada modifiers). 
\[
W_\pm = \slim_{t\to\pm\infty} e^{itH} J_\pm e^{-itH_0}.
\]
Then the existence of these limits are proved by the same method as in the papers 
by Isozaki-Kitada \cite{Isozaki-Kitada1} or Robert \cite{Robert}, and $W_\pm$ are partial isometries 
on $\Ran[E_{I_3}(H_0)]$. Moreover, the asymptotic completeness is also proved 
by the standard method:
\[
\Ran [W_\pm E_{I_3}(H_0)] = E_{I_3}(H)\mathcal{H}_{c}(H),
\]
where $\mathcal{H}_{c}(H)$ is the continuous spectral subspace with respect to $H$. 
The scattering operator $S$ (with essentially a smooth energy cut-off $\i_2(H_0)$) 
is defined by 
\[
S=(W_+)^* W_-,
\]
and it is an isometry on $\Ran[E_{I_3}(H_0)]$. It is well-known that $S$ commutes with 
the free Hamiltonian: $SH_0=H_0 S$. 

We recall a representation formula for the scattering matrix: 
\begin{equation}\label{eq-s-matrix-representation}
S(\l) = -2\pi iT(\l) J_+^* G_- T(\l)^* 
+2\pi i T(\l) G_+^* (H-\l-i0)^{-1} G_- T(\l)^* 
\end{equation}
for $\l\in I$, which is due to Isozaki-Kitada \cite{Isozaki-Kitada2} and Yafaev \cite{Yafaev-1998}. 
We give a proof of the formula in Appendix~\ref{appendix-smatrix-formula} for the completeness. 
The second term in the right hand side is a smoothing operator by virtue of the 
microlocal resolvent estimate of Isozaki-Kitada type \cite{IK0,IK3}. 
The resolvent estimate under our setting is proved in Nakamura \cite{N2017}. 
Thus it remains to compute the first term as a Fourier integral operator. 

We consider the oscillatory integral: 
\begin{align*}
\mathcal{F} J_+^* G_- f(\x) 
&=(2\pi)^{-d} \iiint e^{-i\g_+(x,\x) +i\g_-(x,\y) -iy\cdot\y} \Theta_+(x,\x)\Theta_-(x,\y)\times\\
&\quad \times \overline{a_+(x,\x)} g_-(x,\y) f(y) dyd\y dx, 
\end{align*}
and we compute the integration in $(x,\y)$ using the stationary phase method. 
The stationary phase points are given by 
\begin{equation}\label{eq-stph}
\begin{aligned}
&\pa_x(-\g_+(x,\x)+\g_-(x,\y))= 0, \quad\text{i.e., } \pa_x\g_+(x,\x) =\pa_x \g_-(x,\y), \\
&\pa_\y (\g_-(x,\y)-y\cdot\y)=0, \quad \text{i.e., }\pa_\y \g_-(x,\y)=y. 
\end{aligned}
\end{equation}
Thus these stationary points correspond to the map
\[
\begin{pmatrix} y \\ \y \end{pmatrix}
= \begin{pmatrix} \pa_\y\g_-(x,\y) \\ \y \end{pmatrix}
\underset{w_-}{\longmapsfrom}
\begin{pmatrix} x \\ \pa_x \g_-(x,\y)\end{pmatrix}
= \begin{pmatrix} x \\  \pa_x\g_+(x,\x) \end{pmatrix}
\underset{w_+}{\longmapsto}
\begin{pmatrix} \pa_\x\g_+(x,\x) \\ \x \end{pmatrix}.
\]
These classical wave maps $w_\pm$ are local diffeomorphism, and the composition is also. 
For fixed $(y,\x)$, with $p_0(\x)\in I$, we write the stationary phase points by
\[
x=x(y,\x), \quad \y=\y(y,\x), 
\]
and we set
\[
\g(y,\x) = \g_+(x(y,\x),\x)-\g_-(x(y,\x),\y(y,\x))+y\cdot \y(y,\x)
\]
be the stationary phase. We can show by the construction of 
$\g_\pm$ that $\g(x,\x)-x\cdot\x\in S(\jap{x}^{1-\m},g)$ on $\bigset{(x,\x)}{\b_1<\cos(x,v(\x))<-\b_1}$. 
Then, as is expected, $\g(y,\x)$ is the generating function 
of the classical scattering map : $w_+\circ w_-^{-1}$, i.e., 
\[
\pa_y\g(y,\x) = \y(y,\x), \quad 
\pa_\x\g(y,\x) = \pa_\x\g_+(x(y,\x),\x).
\]
In fact, we have 
\begin{align*}
\pa_y \g(y,\x) 
&=(\pa_y x) \pa_x\g_+(x,\x)-(\pa_y x)\pa_x\g_-(x,\y) \\
&\qquad -(\pa_y\y) \pa_\y\g_-(x,\y)  +(\pa_y\y)y +\y\\
&= (\pa_y x) (\pa_x\g_+(x,\x)-\pa_x\g_-(x,\y)) \\
&\qquad -(\pa_y\y) (\pa_\y\g_-(x,\y) -y) +\y\
=\y
\end{align*}
by the stationary phase equations. Similarly we have 
\begin{align*}
\pa_\x\g(y,\x)
&= \pa_\x\g_+(x,\x) + (\pa_\x x) (\pa_x\g_+(x,\x)-\pa_x\g_-(x,\y))\\
&\qquad -(\pa_\x\y) (\pa_\y\g_-(x,\y) -y)\\
&= \pa_\x\g_+(x,\x).
\end{align*}

In order to apply the stationary phase method, we need to compute the Hessian at the stationary 
phase points: 

\begin{lem}\label{lem-qm-Hess}
Let $\mathrm{Hess}(y,\x)$ be the Hessian of $-\g_+(x,\x)+\g_-(x,\y)-y\cdot\y$ with respect to 
$(x,\y)$ at the stationary points. Then 
\begin{equation*}
\mathrm{Hess} = (-1)^d \det(\pa_x\pa_\x\g_-(x,\y))\det(\pa_x\pa_\x\g_+(x,\x))
\det(\pa_y\pa_\x\g(y,\x))^{-1}.
\end{equation*}
\end{lem}

\begin{proof}
We compute 
\begin{align*}
\mathrm{Hess} 
&= \det((\pa_x,\pa_\y)^2 (-\g_+(x,\x)+\g_-(x,\y)-y\cdot\y))\big|_{x=x(y,\x),\y=\y(y,\x)}\\
&= \det \begin{pmatrix} -\pa_x\pa_x\g_+(x,\x)+\pa_x\pa_x\g_-(x,\y) & \pa_x\pa_\y \g_-(x,\y) \\
\pa_\y\pa_x\g_-(x,\y) & \pa_\y\pa_\y \g_-(x,\y) 
\end{pmatrix}\bigg|_{x=x(y,\x),\y=\y(y,\x)}.
\end{align*}
It is easy to see 
\begin{align*}
&\begin{pmatrix} -\pa_x\pa_x\g_++\pa_x\pa_x\g_- & \pa_x\pa_\y \g_- \\
\pa_\y\pa_x\g_- & \pa_\y\pa_\y \g_-
\end{pmatrix}\begin{pmatrix} E & 0 \\
(\pa_x\pa_\y \g_-)^{-1}(\pa_x\pa_x\g_+-\pa_x\pa_x\g_-) & E
\end{pmatrix}\\
&=\begin{pmatrix} 0 & \pa_x\pa_\y \g_- \\
\pa_\y\pa_x\g_-+(\pa_\y\pa_\y\g_-)(\pa_x\pa_\y \g_-)^{-1}
(\pa_x\pa_x\g_+-\pa_x\pa_x\g_-)  & \pa_\y\pa_\y \g_-
\end{pmatrix},
\end{align*}
and hence 
\[
\mathrm{Hess} 
=(-1)^d \det (\pa_x\pa_\y \g_-)\det\bigpare{
\pa_\y\pa_x\g_-+(\pa_\y\pa_\y\g_-)(\pa_x\pa_\y \g_-)^{-1}
(\pa_x\pa_x\g_+-\pa_x\pa_x\g_-)}.
\]
Now we differentiate the stationary phase equation \eqref{eq-stph} in $y$ to learn 
\begin{align}
&(\pa_y x)\pa_x\pa_x\g_+ = (\pa_y x)\pa_x\pa_x\g_-
+(\pa_y \y)\pa_\y\pa_x \g_-, \label{eq-sp-diff-1}\\
& (\pa_y x)\pa_x\pa_\y\g_-+(\pa_y \y)\pa_\y\pa_\y\g_- = E.\label{eq-sp-diff-2}
\end{align}
From \eqref{eq-sp-diff-1}, we have 
\[
\pa_y\y = (\pa_y x)(\pa_x\pa_x\g_+-\pa_x\pa_x\g_-)(\pa_\y\pa_x\g_-)^{-1}.
\]
Substituting this to \eqref{eq-sp-diff-2}, we have 
\[
(\pa_y x)(\pa_x\pa_\y\g_- + (\pa_x\pa_x\g_+-\pa_x\pa_x\g_-)(\pa_\y\pa_x\g_-)^{-1}\pa_\y\pa_\y\g_-)
=E,
\]
and hence 
\[
(\pa_y x)^{-1} = \pa_x\pa_\y\g_- + (\pa_x\pa_x\g_+-\pa_x\pa_x\g_-)(\pa_\y\pa_x\g_-)^{-1}\pa_\y\pa_\y\g_-,
\]
or 
\[
{}^t(\pa_y x)^{-1} = 
\pa_\y\pa_x\g_-+(\pa_\y\pa_\y\g_-)(\pa_x\pa_\y \g_-)^{-1}
(\pa_x\pa_x\g_+-\pa_x\pa_x\g_-).
\]
Substituting this to the above formula on the Hessian, we learn 
\[
\mathrm{Hess} = \det(\pa_x\pa_\y\g_-(x,\y))\cdot\det(\pa_y x(y,\x))^{-1}, 
\]
where $x=x(y,\x)$, $\y=\y(y,\x)$. If we set
\[
z(y,\x)= (\pa_\x\g_+)(x(y,\x),\x),
\]
then, since $\g$ is the generating function of $w_+\circ w_-^{-1}$, we learn 
\[
\pa_y \pa_\x\g(y,\x)= \pa_y z = (\pa_y x)\cdot(\pa_x\pa_\x\g_+)(x,\x).
\]
Combining these, we conclude the assertion.
\end{proof}

Now we denote $x(y,\x)$ be the stationary point as above, and denote the corresponding momentum 
at $t=0$ by 
\begin{equation}\label{eq-def-zeta}
\z(y,\x) = \pa_x\g_-(x(y,\x),\y(y,\x)) = \pa_x\g_+(x(y,\x),\x).
\end{equation}
We also denote 
\[
\Theta(y,\x)= \Bigabs{\det\Bigpare{\tfrac{\pa^2\g}{\pa y\pa\x}(y,\x)}}^{1/2}.
\]
Then using the stationary phase method and the standard oscillatory integral calculation, 
we have the following expression of $J_+^* G_-$ (see, e.g., Asada-Fujiwara \cite{Asada-Fujiwara} Section~3). 

\begin{lem} \label{lem-JG-formula} There is $Z(x,\x)\in S(\jap{x}^{-1},g)$ such that 
\[
\mathcal{F} J_+^*G_- f(\x) =(2\pi)^{-d/2}\int e^{-i\g(y,\x)}\Theta(y,\x)  Z(y,\x)f(y)dy. 
\]
Moreover, $Z$ is essentially supported in 
\[
\O= \bigset{(y,\x)}{p_0(\x)\in I, \cos(x(y,\x),v(\z(y,\x)))\in [-\b_2,-\b_1]},
\]
i.e., for any $\a,\b\in\ze_+^d$ and $N\geq 0$, 
\[
\bigabs{\pa_y^\a\pa_\x^\b Z(y,\x)}\leq C_{\a\b N}\jap{y}^{-N},
\quad\text{for }(y,\x)\notin \O.
\]
The principal symbol of $Z(y,\x)$ is given by 
\[
Z_0(y,\x)= \overline{a_+(x(y,\x),\x)}g_-(x(y,\x),\pa_y\g(y,\x)),
\]
i.e., $Z-Z_0\in S(\jap{x}^{-2},g)$. 
\end{lem}

In order to compute $T(\l)J_+^* G_- T(\l)^*$, we note the following basic property 
of the generating function $\g(y,\x)$, which essentially says $\g(y,\x)$ restricted to 
$\Sigma_\l$ defines a canonical map on $T^*\Sigma_\l$. 
We recall that by the energy conservation, we have 
\[
p_0(\pa_y\g(y,\x)) = p_0(\x), \quad \x\in p_0^{-1}(I).
\]

\begin{lem}\label{lem-psi-invariance}
For $p_0(\x)=\l\in I$, 
\[
\g(y+tv(\pa_y\g(y,\x)),\x) =\g(y,\x), \quad t\in\re. 
\]
\end{lem}

\begin{proof}
We choose a local coordinate near $\Sigma_\l$ such that 
$p_0(\x)=\l+\x_1$ and hence  
\[
\Sigma_\l =\bigset{(0,\x')}{\x'\in\re^{d-1}}, \quad v(\x)=\pa_\x p_0(\x)=(1,0,\dots,0)
\]
in the neighborhood. We may suppose $\x$ and $\pa_y\g(y,\x)$ are contained in the 
neighborhood, and hence $v(\x)=v(\pa_y\g(y,\x))=(1,0,\dots,0)$. 
We note,  since $\pa_y\g(y,\x)\in \Sigma_\l$, $\pa_{y_1}\g(y,\x) =0$ in this coordinate. 
Thus we have 
\begin{align*}
\pa_t \g(y+tv(\pa_y\g(y,\x)),\x) 
&= v(\pa_y\g(y,\x))\cdot\pa_y\g(y+tv(\pa_y\g(y,\x)),\x) \\
&= \pa_{y_1}\g(y+t(1,0,\dots,0),\x) =0.
\end{align*}
This implies the assertion. 
\end{proof}

In the following, we consider Fourier integral operators defined on $\Sigma_\l$, 
and here we introduce several notations. We usually work in a local coordinate 
in $\Sigma_\l$, and since we are interested in the behavior of operators/symbols 
for large $|x|$, and hence we may suppose $\x$, $\pa_y\g(y,\x)$, $\pa_x\g_\pm(x,\x)$, 
etc., are in the same local coordinate patch. For $\x\in\Sigma_\l$, we identify 
the cotangent space at $\x$: $T^*_\x\Sigma_\l$ with  $v(\x)^\perp$, i.e., 
the orthogonal subspace of the normal vector $v(\x)=\pa_\x p_0(\x)$, as usual. 
We employ the standard metric on $T^*_\x\Sigma_\l$. For $a(x,\x)\in C^\infty(T^*\Sigma_\l)$, 
we write $a\in S(m(x,\x),\tilde g)$ if for any multi-indices $\a,\b\in\ze_+^{d-1}$, 
\[
\bigabs{\pa_x^\a\pa_\x^\b a(x,\x)}\leq C_{\a\b}\jap{x}^{-|\a|}m(x,\x), 
\quad  x\in\re^{d-1}, \x\in \Sigma_\l,
\]
in the local coordinate. We note it is not always natural to consider $x\in T^*_\x\Sigma_\l$ 
in the above expression, since we consider Fourier integral operators, and hence 
$x$ may be better to be considered as an element in another cotangent space. 
In our case, here we consider in a local coordinate patch, and the condition is 
well-defined without ambiguities. By virtue of Lemma~\ref{lem-psi-invariance}, we may define 
\[
\tilde \g(y,\x) =\g(y,\x), \quad \tilde\Theta(y,\x)=\Theta(y,\x)
\]
on $T^*\Sigma_\l$ using the local coordinate, where $y$ should be considered as an element of 
$T^*_\y \Sigma_\l$ with $\y=\pa_y\g(y,\x)$. 

We compute the operator $T(\l)J_+^* G_-T(\l)^*$ using the local coordinate 
in the above proof. Then, as well as in the proof of Lemma~5.4 of \cite{N2016}, 
for $f\in C_0^\infty(\Sigma_\l)$ supported in the neighborhood, we have 
\begin{align}
&T(\l)J_+^* G_-T(\l)^*f(\x') \nonumber \\
&=c_d \iint e^{-i\g(y,(0,\x'))+iy\cdot(0,\y')}\Theta(y,(0,\x'))Z(y,(0,\x'))f(\y')d\y' dy
\nonumber\\
&=c_{d-1} \iint \biggpare{\frac{1}{2\pi}\int_{-\infty}^\infty e^{-i\g((t,y'),(0,\x'))+iy'\cdot\y'}
\times  \\&\qquad \times 
\Theta((t,y'),(0,\x'))Z((t,y'),(0,\x'))dt } f(\y') d\y' dy' \nonumber\\
&= c_{d-1} \iint \biggpare{\frac{1}{2\pi}\int_{-\infty}^\infty e^{-i\tilde\g(y',\x')+iy'\cdot\y'}
\tilde \Theta(y',\x') Z((t,y'),(0,\x'))dt } 
f(\y')d\y' dy'\nonumber \\
&= c_{d-1}\iint \biggpare{\frac{1}{2\pi} \int_{-\infty}^\infty Z((t,y'),(0,\x'))dt } 
e^{-i\tilde\g(y',\x')+iy'\cdot\y'}\tilde \Theta(y',\x')f(\y') d\y' dy', 
\label{eq-TJGT-formula}
\end{align}
where $c_\ell =(2\pi)^{-\ell}$.
Thus we formally observe that 
$T(\l)J_+^* G_-T(\l)^*$ is a Fourier integral operator on $\Sigma_\l$ with the 
phase function $\tilde\g(y',\x')$. 
In other words, we have 
\begin{equation}\label{eq-qm-mainpart}
\begin{split}
T(\l)J_+^* G_-T(\l)^*f(\x)
&= c_{d-1}\iint \biggpare{\frac{1}{2\pi} \int_{-\infty}^\infty Z(y+tv(\pa_y\g(y,\x)), \x)dt} 
\times \\
&\qquad \times 
e^{-i\tilde\g(y,\x)+iy\cdot\y}\tilde \Theta(y,\x)f(\y) d\y dy
\end{split}
\end{equation}
on $T^*\Sigma_\l$. It remains to compute the symbol and thus justify the computation. 

\begin{lem}\label{lem-Z-formula}
\[
\int_{-\infty}^\infty Z(y+tv(\pa_y\g(y,\x)),\x)dt = i +R(y,\x)
\]
on $T^*\Sigma_\l$, where $R\in S(\jap{x}^{-1},\tilde g)$. 
\end{lem}

\begin{proof}
We fix $(y_0,\x_0)$, and let 
\[
z_0=x(y_0,\x_0), \quad \z_0=\z(y_0,\x_0)
\]
be the stationary phase points as in \eqref{eq-def-zeta}. 
We also write $\y_0=\pa_y\g(y_0,\x_0)$. 
Then by the construction, we observe 
\[
w_-(z_0,\z_0)= (y_0,\y_0), \quad \text{or equivalently,}\quad  (z_0,\z_0)=w_-^{-1}(y_0,\y_0).
\]
We note 
\[
\exp tH_{p_0} (y_0,\y_0) = (y_0+t v(\y_0), \y_0), 
\]
and combining this with the intertwining property: 
\[
\exp tH_p\circ w_-^{-1} = w_-^{-1}\circ \exp tH_{p_0}, 
\]
we learn 
\[
w_-^{-1}(y_0+tv(\y_0),\y_0) =\exp tH_p (z_0,\z_0),
\]
and hence 
\[
(x(y_0+tv(\y_0),\x_0),\z(y_0+tv(\y_0),\x_0)  = \exp tH_p (z_0,\z_0). 
\]
We denote
$(z(t),\z(t))=\exp tH_p(z_0,\z_0)$ as in the last section.
We note, by Lemmas~\ref{lem-ik-G-formula} and \ref{lem-JG-formula}, 
the principal symbol of $Z(y,\x)$ is given by
\[
Z_{00}(y,\x)= -i(\pa_\x p)(x,\pa_x\g_-(x,\y))\cdot\pa_x\i_-(x,\y). 
\]
with $x=x(y,\x)$ and $\y=\pa_y\g(y,\x)$. Hence we have 
\[
Z_{00}(y_0+tv(\y_0),\x_0) = -i(\pa_\x p)(z(t),\z(t))\cdot(\pa_x\i_-)(z(t),\y_0), 
\quad t\in\re.
\]
By the Hamilton equation, we note $(\pa_\x p)(z(t),\z(t))=\frac{d}{dt} z(t)$, and hence 
\[
Z_{00}(y_0+tv(\y_0),\x_0) = -i\frac{d}{dt}(\i_-(z(t),\y_0)).
\]
Since $\lim_{t\to\infty}\i_-(z(t),\y_0)=0$ and $\lim_{t\to-\infty}\i_-(z(t),\y_0)=1$, 
we have 
\begin{align*}
\int_{-\infty}^\infty Z_{00}(y_0+tv(\y_0),\x_0)dt &= -i \int_{-\infty}^\infty \frac{d}{dt}\i_-(z(t),\y_0)dt\\
&= -i\Bigpare{\lim_{t\to\infty} \i_-(z(t),\y_0)-\lim_{t\to-\infty} \i_-(z(t),\y_0)} = i.
\end{align*}

Now it remains to estimate the contribution from the lower order term: 
$R(y,\x)= Z(y,\x)-Z_{00}(y,\x)\in S(\jap{x}^{-2},g)$. 

As usual, we identify $T_\x^*\Sigma_\l$ with $v(\x)^\perp$, the orthogonal subspace 
of the normal vector $v(\x)$ at $x\in \Sigma_\l$. Then $y_0\perp \y_0$ and hence 
\[
|y_0+t\y_0|=(|y_0|^2+t^2 |v(\y_0)|^2)^{1/2}\geq c_0(|y_0|+t|v(\y_0)|), 
\]
where $c_0=1/\sqrt{2}$. Thus we have 
\[
|R(y_0,\x_0)|\leq C \int_{-\infty}^\infty (1+|y_0|+t|v(\y_0)|)^{-2}dt \leq C' \jap{y_0}^{-1}, 
\]
where $\x_0\in \Sigma_\l$, $y_0\in T^*_{\y_0}\Sigma_\l$. 
Similarly, we can show, for any $\a,\b\in\ze_+^{d-1}$, 
\[
\bigabs{\pa_y^\a\pa_\x^\b R(y,\x)}\leq C_{\a\b}\jap{y}^{-1-|\a|}, 
\]
which completes the proof. 
\end{proof}

Thus we learn, combining the lemma with \eqref{eq-qm-mainpart}, 
\[
T(\l)J_+^* G_-T(\l)^*f(\x)= \frac{c_{d-1}}{2\pi} \iint (i+R(y,\y))
e^{-i\tilde\g(y,\x)+iy'\cdot\y}\tilde \Theta(y,\x)f(\y) d\y dy
\]
with $R\in S(\jap{x}^{-1},\tilde g)$. Substituting this to the representation formula, 
\eqref{eq-s-matrix-representation}, we obtain  
\[
S(\l)f(\x)=  c_{d-1}\iint (1-iR(y,\y))
e^{-i\tilde\g(y,\x)+iy'\cdot\y}\tilde\Theta(y,\x)f(\y) d\y dy. 
\]
This complete the proof of Theorem~\ref{thm-main}. \qed 

\appendix


\section{Representation formula of the scattering matrix}
\label{appendix-smatrix-formula}

In this appendix, we sketch the proof of \eqref{eq-s-matrix-representation}. 
We suppose $f, g\in\mathcal{S}(\re^d)$ such that $\hat f,\hat g\in C_0^\infty(p_0^{-1}(I))$, and 
we write $f(\l)= T(\l)f$, $g(\l)=T(\l)g$, $\l\in\ I$. We first note, by the standard Cook-Kuroda method, we have
\begin{equation}\label{eq-app-smatrix-1}
W_\pm^I f = J_\pm f +i\int_0^{\pm\infty} e^{itH} G_\pm e^{-itH_0} f dt.
\end{equation}
We also note, by the construction of $J_\pm$, 
\[
\bignorm{J_\pm e^{-itH_0} f }\to 0, \quad \text{as }t\to\mp\infty,
\]
and hence 
\begin{align}
W_\pm^I f &= \lim_{t\to\pm\infty} \Bigpare{e^{itH} J_\pm e^{-itH_0} f
-e^{-itH} J_\pm e^{itH_0} f}  \nonumber\\
&= \pm i \int_{-\infty}^\infty e^{itH} G_\pm e^{-itH_0} f dt.
\label{eq-app-smatrix-2}
\end{align}
Using \eqref{eq-app-smatrix-2}, we compute
\begin{align*}
\bigjap{f, S^I g} &= \bigjap{W_+^I f, W_-^I g} \\
&= -i \int_{-\infty}^\infty \bigjap{W_+^I f, e^{itH} G_- e^{-itH_0} g} dt\\
&= -i \int_{-\infty}^\infty \bigjap{W_+^I e^{-itH_0}f, G_- e^{-itH_0} g} dt. 
\end{align*}
In the last line, we have used the intertwining property. 
Then we substitute \eqref{eq-app-smatrix-1} to learn 
\begin{align*}
\bigjap{f, S^I g} &= -i \int_{-\infty}^\infty \bigjap{J_+ e^{-itH_0}f, G_- e^{-itH_0} g} dt \\
&\qquad - \int_0^\infty  \int_{-\infty}^\infty \bigjap{e^{isH} G_+ e^{-i(s+t)H_0} f, 
G_- e^{-itH_0} g} dt\, ds.
\end{align*}
Now we use the spectral representation:
\[
f= \int_I T(\l)^* f(\l)d\l, \quad g=\int_I T(\s)^* g(\s)d\s
\]
to obtain (at least formally)
\begin{align*}
&\bigjap{f, S^I g} = -i \int_{-\infty}^\infty dt  \int_I d\l \int_I d\s 
\bigjap{J_+ e^{-itH_0}T(\l)^*f(\l), G_- e^{-itH_0} T(\s)^*g(\s)} \\
&- \int_0^\infty ds \int_{-\infty}^\infty dt \int_I d\l \int_I d\s 
\bigjap{e^{isH} G_+ e^{-i(s+t)H_0} T(\l)^* f(\l), G_- e^{-itH_0} T(\s)^* g(\s)} \\
&= -i \int_{-\infty}^\infty dt  \int_I d\l \int_I d\s 
e^{it(\l-\s)} \bigjap{J_+ T(\l)^*f(\l), G_-  T(\s)^*g(\s)} \\
& - \int_0^\infty ds \int_{-\infty}^\infty dt \int_I d\l \int_I d\s 
e^{it(\l-\s)} \bigjap{e^{is(H-\l)} G_+ T(\l)^* f(\l), G_- T(\s)^* g(\s)}.
\end{align*}
Here we note that 
\[
\mathrm{WF}(\mathcal{F} T^*(\l)f(\l))
\subset \bigset{(\x,x)}{\x\in\Sigma_\l, x\perp \Sigma_\l =\re v(\x)},
\]
and the essential support of the amplitudes of $G_\pm$ are disjoint from it. 
Hence 
\[
G_+T(\l)^*f(\l), G_-(\s)T(\s)^*g(\s)\in \mathcal{S}(\re^d),
\]
and 
these integrants are well-defined, smooth in the parameters. Thus we can change the order of 
integration, and using the formula
\[
\int_{-\infty}^\infty e^{it(\l-\s)}dt = \d(\l-\s)
\]
in the distribution sense, we learn 
\begin{align*}
\bigjap{f, S^I g} &= -2\pi i  \int_I d\l 
\bigjap{ T(\l)^*f(\l), J_+^* G_-  T(\l)^*g(\l)} \\
&\quad  - 2\pi \int_0^\infty ds\int_I d\l 
 \bigjap{ T(\l)^* f(\l), G_+^* e^{-is(H-\l)} G_- T(\l)^* g(\l)}.
\end{align*}
By the microlocal resolvent estimate \cite{N2016}, we learn that 
\[
\int_0^\infty G_+^* e^{-is(H-\l)} G_- ds = -i G_+^* (H-\l-i0)^{-1} G_- 
\]
makes sense, and we conclude 
\[
\bigjap{f, S^I g} = -2\pi i  \int_I d\l 
\bigjap{ T(\l)^*f(\l), \bigpare{J_+^* G_- -G_+^*(H-\l-i0)^{-1}G_-} T(\l)^*g(\l)}.
\]
This implies \eqref{eq-s-matrix-representation}. \qed


\end{document}